\documentclass[11pt]{article}
\usepackage{geometry}                
\usepackage{graphicx}

\usepackage[utf8]{inputenc} 
\usepackage{textcomp} 

\usepackage{flafter}  

\usepackage{amsmath,amssymb}  
\usepackage{bm}  
\usepackage{amsthm}

\usepackage{memhfixc}  

\newtheorem{definition}{Definition}[section]
\newtheorem{theorem}{Theorem}[section]
\newtheorem{proposition}{Proposition}[section]
\newtheorem{corollary}{Corollary}[section]
\newtheorem{remark}{Remark}[section]
\newtheorem{lemma}{Lemma}[section]
\numberwithin{equation}{section}

\title{Dual of the Hopf Algebra Consisting of the Adjacency Matrices}
\author{Zhou Mai \footnote{Address: Colleague of Mathematical Science, Nankai University, Weijin Road, Tianjin City, Republic China; Email address: zhoumai@nankai.edu.cn}}

\begin{document}

\maketitle

\begin{abstract}
In this article we discuss the Hopf algebras spanned
by the adjacency matrices in detail. We show that there
two Hopf algebraic structures concerning the adjacency
matrices, one is the copy of Connes-Kreimer Hopf algebra,
another one is the copy of the dual of Connes-Kreimer
Hopf algebra. 
\end{abstract}

\tableofcontents

\section{Introduction}

It is well known that the adjacency matrices
indecate the multigraphs (see \cite{1}) which
can be regarded as Feynman diagrams without
external lines. To indecate the general Feynman 
diagrams with the external lines, we introduce 
the notation of the extended adjacency matrices.
In the present article we discuss the Hopf 
algebras over $\mathbb{C}$ spanned by
the set of the all adjacency matrices. More 
precisely, the vector spaces under consideration
denoted by $\mathcal{H}_{adj}$ (or  
$\mathcal{H}_{adj(e)}$ in the situation of the 
extended adjacency matrices)
are the ones spanned by the equivalent classes 
of the adjacency matrices. The equivalent relations 
are usual and natural ones (\cite{1}) to describe
the isomorphic classes of the graphs (or of Feynman
diagrams). Due to the correspondence between
the adjacency matrices and Feynman diagrams (\cite{1,5}),
the vector spaces in our setting is another
version of the ones in Connes-Kreimer theory (\cite{2,3,4}).

We prove that there are two Hopf 
algebraic structures on $\mathcal{H}_{adj}$
(or on $\mathcal{H}_{adj(e)}$). The first Hopf
algebra denoted by $(\mathcal{H}_{adj},\oplus,
u,\bigtriangleup,\eta,S)$ (or $(\mathcal{H}_{adj(e)},
\oplus,u,\bigtriangleup,\eta,S)$) is the copy of 
Connes-Kreimer Hopf algebra (\cite{2,3,4}). 
The commutative multiplication $\oplus$ is reduced from
the direct sum of the matrices corresponding to
the disjoint union of the graphs. 
The coproduct $\bigtriangleup$ is defined in 
terms of the quotient which is the copy of the 
quotient of Feynman diagrams. $u$ and $\eta$
are the unit and the co-unit respectivly.
$S$ is the antipode. In this article we focus on
the second Hopf algebra
denoted by $(\mathcal{H}_{adj},\bullet,
u,\bigtriangleup_{1},\eta,S_{1})$
(or $(\mathcal{H}_{adj(e)},\bullet,
u,\bigtriangleup_{1},\eta,S_{1})$) which is isomorphic
to the dual hopf algebra of $(\mathcal{H}_{adj},\oplus,
u,\bigtriangleup,\eta,S)$. The multiplication
$\bullet$ in $(\mathcal{H}_{adj},\bullet,u,
\bigtriangleup_{1},\eta,S_{1})$ is defined 
with the help of the notion of the insertion
which is the copy of the insertion of Feynman
diagrams (\cite{2,3,4}). We detail the multiplication
$\bullet$ and the coproduct $\bigtriangleup_{1}$. 
Moreover, the structure of $(\mathcal{H}_{adj},
\bullet,u,\bigtriangleup_{1},\eta,S_{1})$
is described in a explicit way. The unit $u$ and
the co-unit $\eta$ of $(\mathcal{H}_{adj},
\bullet,u,\bigtriangleup_{1},\eta,S_{1})$ are 
same as ones of $(\mathcal{H}_{adj},\oplus,
u,\bigtriangleup,\eta,S)$. Because $\oplus$
is commutative, $\bigtriangleup_{1}$ is
co-commutative. Both $\bigtriangleup$
and $\bigtriangleup_{1}$ are conilpotent,
therefore the antipodes $S$ and $S_{1}$
can be given by the standard formula concerning
the products and reduced coproducts (\cite{6}).

The present paper is organized as follows.
In the section 2 we discuss the Hopf algebra
consisting of the adjacency matrices which
is a different version of Connes-Kreimer
Hopf algebra by means of the matrix. At
beginning of this section we talk about 
some basic subjects concerning the adjacency 
matrices (or the extended adjacency matrices),
for example, the equivalent relation,
the direct sum and the connectivity. Then
we discuss the quotient of the adjacency
matrices which is parallel to the quotient
of Feynman diagrams in Connes-Kreimer
theory. In addition, based on the notation
of the quotient, we can define the coproduct
on $\mathcal{H}_{adj}$, or on $\mathcal{H}_{adj(e)}$,
such that they become the Hopf algebras
$(\mathcal{H}_{adj},\oplus,
u,\bigtriangleup,\eta,S)$ or $(\mathcal{H}_{adj(e)},
\oplus,u,\bigtriangleup,\eta,S)$.
In the section 3 we consider the insertion
of the adjacency matrices, or the extended
adjacency matrices, which can be regarded as
the translation of the insertion of Feynman
diagrams into the language of the matrix. The
properties of the insertion are discussed in
detail. In the section 4 we turn to 
the dual of $(\mathcal{H}_{adj},\oplus,
u,\bigtriangleup,\eta,S)$ (or $(\mathcal{H}_{adj(e)},
\oplus,u,\bigtriangleup,\eta,S)$). We prove
that the dual of $(\mathcal{H}_{adj},\oplus,
u,\bigtriangleup,\eta,S)$ can be realized
on $\mathcal{H}_{adj}$, i.e. there is a Hopf
algebra $(\mathcal{H}_{adj},\bullet,
u,\bigtriangleup_{1},\eta,S_{1})$ being
isomorphic to the dual of $(\mathcal{H}_{adj},\oplus,
u,\bigtriangleup,\eta,S)$. The product $\bullet$
and the coproduct $\bigtriangleup_{1}$ are
described in detail. Moreover we have
$\mathcal{H}_{adj}=U(\mathbf{P}(\mathcal{H}_{adj}))$,
where $\mathbf{P}(\mathcal{H}_{adj})$ is the
Lie algebra consisting of primitive elements
of $(\mathcal{H}_{adj},\bullet,
u,\bigtriangleup_{1},\eta,S_{1})$. The situation
of the extended adjacency matrices is similar.

\section{Hopf algebras of adjacency matrices}

In this section we will discuss the 
Hopf algebra consisting of the adjacency matrices.
For simplification we focus on the adjacency
matrices with zero diagonal. The general
situation is similar. To indecate Feynman
diagrams with external lines, we introduce
the notation of the extended adjacency matrices
which are also the adjacency matrices divided 
into internal part and external part. 
Actually, a more general situation, 
the complex matrices with zero diagonal,
was discussed in \cite{7}.

\subsection{The basic notations and the connectivity 
about the adjacency matrices}

At the beginning of this subsection we introduce some
notations. In this article we set $[m]=\{1,\cdots,m\}$
for a positive integer $m$. For a finite
set $I$, we let $|I|$ denote the number of the
elements in $I$, and $\mathbf{Part}(I)$ denotes the
set of all partitions of $I$, i.e.

$$
\mathbf{Part}(I)=\{\{I_{i}\}_{i=1}^{k}|I_{i}\subset I,
I=\bigcup_{i=1}^{k}I_{i},I_{i}\cap I_{i^{\prime}}
=\emptyset,i\not=i^{\prime},1\leq i,i^{\prime}\leq k,
k\leq|I|\}.
$$
The symbol $\mathbf{part}(I)$ denotes the set
of all sequences of disjoint subsets in $I$, i.e.

$$
\mathbf{part}(I)=\{\{I_{i}\}_{i=1}^{k}|
\{I_{i}\}_{i=1}^{k}\in\mathbf{Part}
(\bigcup_{i=1}^{k}I_{i}),
\bigcup_{i=1}^{k}I_{i}\subset I\}.
$$
For two sequences of the disjoint subsets
$\{I_{i}\},\{J_{j}\}\in\mathbf{part}(I)$, 
we say $\{I_{i}\}\subset\{J_{j}\}$, if for
each $I_{i}$ there is a $J_{j}$ such that
$I_{i}\subset J_{j}$.

We now turn to the discussion of the adjacency matrices.

\begin{definition}

\begin{itemize}
\item An adjacency matrix
is a symmetric matrix with non-negative integer entries
and zeros along the main diagonal. 
We call $\sum_{i<j}m_{ij}$ the
degree of $M$ denoted by $\mathbf{deg}M$.  
The set of adjacency matrices of
$m\times m$ is denoted by $M_{adj}(m,\mathbb{N})$.
\item Let $M\in M_{adj}(m,\mathbb{N})$ be an
adjacency matrix, 
$a=(a_{1},\cdots,a_{m})\in\mathbb{N}^{m}$ be a 
multiple index. Then, an extended adjacency 
matrix $(M,a)$ is defined to be an adjacency 
matrix of order $m+1$ with the following form,

\begin{equation}
(M,a)=
\begin{pmatrix}
M & a^{T} \\
a & o
\end{pmatrix},
\end{equation}
where $M$ is called the internal part of
$(M,a)$, and $a$ is called the external part
of $(M,a)$. The degree of an extended 
adjacency matrix $(M,a)$ is same as one 
of its internal part, i.e.
$\mathbf{deg}(M,a)=\mathbf{deg}M$.
The set of the all extended adjacency of order
$m+1$ is denoted by $M_{adj}(m+1,\mathbb{N})_{(e)}$.
\end{itemize}	
\end{definition}

\begin{remark}
An adjacency matrix
$M\in M_{adj}(m,\mathbb{N})$ indecates a
Feynman diagram without
external lines and loops, or a graph without
loops. For an extended adjacency
matrix $(M,b)$, $b$ indecates $|b|=b_{1}+
\cdots+b_{m}$ external lines, where $i$th
vertex of the Feynman diagram is assigned 
to $b_{i}$ external lines ($i=1,\cdots,m$).
\end{remark}

\begin{proposition}
Under the addition of the matrices,
$M_{adj}(m,\mathbb{N})$ is a monoid with generators
$\{M(i,j)\}$, where $M(i,j)=(m_{kl})_{m\times m}$ 
satisfies $m_{kl}=m_{lk}=\delta_{ik}\delta_{jl}$, 
$i\leq j,\,k\leq l$.
\end{proposition}

Recalling every row and every column of a permutation 
matrix contain exactly one nonzero entry, which is 1.
Now we define a equivalent relation on $M_{adj}(m,\mathbb{N})$ 
as follows. Let $M_{1},M_{2}\in M_{adj}(m,\mathbb{N})$, then 

\begin{equation}
M_{1}\sim M_{2}\,\Longleftrightarrow\,
M_{1}=PM_{2}P^{T},
\end{equation}
where $P$ is a permutation matrix.
The equivalent relation mentioned above
can be described in a different way. 
Let $M=(m_{ij})_{m\times m}\in M_{adj}(m,\mathbb{N})$,
$\pi\in \textbf{S}_{m}$ be a permutation
$\pi:\{1,\cdots,m\}\to\{1,\cdots,m\}$,

$$
\pi=
\begin{pmatrix}
1 & 2 & \cdots & m \\
\pi(1) & \pi(2) & \cdots & \pi(m)
\end{pmatrix}.
$$
Then, the action of $\pi$ on $M$
is defined to be an adjacency matrix
$\pi(M)=(m_{ij}^{\prime})_{m\times m}$
satisfying $m_{ij}^{\prime}=m_{\pi(i)\pi(j)}$.
Let $M_{1},M_{2}\in M_{adj}(m,\mathbb{N})$, 
then $M_{1}\sim M_{2}$ if and only if
there is a $\pi\in \textbf{S}_{m}$ such that 
$M_{1}=\pi(M_{2})$. Thus, the equivalent classes under
above equivalent relation are the orbits
of the permutation group $\mathbf{S}_{m}$
acting on $M_{adj}(m,\mathbb{N})$.
Let $M\in M_{adj}(m,\mathbb{N})$, we denote
the equivalent class of $M$, or an orbit 
of $M$, by $\{M\}$, then, 
$\{M\}=\{\pi(M)|\pi\in \textbf{S}_{m}\}$.
The set of equivalent class is denoted
by $M_{adj}(m,\mathbb{N})\diagup{\sim}$.
It is obvious that $\mathbf{deg}M=
\mathbf{deg}(PMP^{T})$, where $P$ is
a permutation matrix. Thus we define
$\mathbf{deg}\{M\}=\mathbf{deg}M$.
We will mainly focus on the equivalent classes
,or orbits, below. 

The equivalent relation concerning the adjacency
matrices can be generalized to the situation
of the extended adjacency matrices. Let

$$
(M_{i},b_{i})=
\begin{pmatrix}
M_{i} & b_{i}^{T} \\
b_{i} & 0
\end{pmatrix}
\in M_{adj}(m+1,\mathbb{N})_{e}
$$
be two extended adjacency matrices of 
oeder $m+1$ ($i=1,2$), we say $(M_{1},b_{1})$
is equivalent to $(M_{2},b_{2})$ if and only if
there is a permutation matrix $P$ of order $m$
such that

$$
\begin{pmatrix}
M_{1} & b_{1}^{T} \\
b_{1} & 0
\end{pmatrix}
=
\begin{pmatrix}
P & 0 \\
0 & 1
\end{pmatrix}
\begin{pmatrix}
M_{2} & b_{2}^{T} \\
b_{2} & 0
\end{pmatrix}
\begin{pmatrix}
P^{T} & 0 \\
0 & 1
\end{pmatrix}.
$$
Let $(M,b)\in M_{adj}(m+1,\mathbb{N})_{(e)}$,
$\pi\in\mathbb{S}_{m}$, we define
$\pi((M,b))=(\pi(M),\pi(b))$, where
$\pi(b)=(b_{\pi(1)}.\cdots,b_{\pi(m)})$.
Similar to the previous situation,
we consider the equivalent class

$$
\{(M,b)\}=\{\pi((M,b))|\pi\in\mathbb{S}_{m}\}.
$$
Thus, each equivalent class is the orbit
of the action of $\mathbb{S}_{m}$.

Let

\begin{equation}
M_{adj}(+\infty,\mathbb{N})=
(\bigcup\limits_{m\geq 2}(M_{adj}(m,\mathbb{N})
\diagup\sim)\setminus\{0\})\cup\{0\}.
\end{equation}
In $M_{adj}(+\infty,\mathbb{N})$, we do not
distinguish the zero matrices with different order.
Actually, from the viewpoint of the graphic
theory, zero matrix corresponding to the
empty set.

Let $M_{i}\in M_{adj}(m_{i},\mathbb{N})$
($i=1,2$), then direct sum $M_{1}\oplus M_{2}
\in M_{adj}(m_{1}+m_{2},\mathbb{N})$. Actually,
the direct sum $M_{1}\oplus M_{2}$ can be 
realized by a block diagonal matrix 

$$
M_{1}\oplus M_{2}=
\mathbf{diag}(M_{1}, M_{2})=
\begin{pmatrix}
M_{1} & 0 \\
0 & M_{2}
\end{pmatrix}.
$$

The direct sum mentioned above can be extened into
$M_{adj}(+\infty,\mathbb{N})$.  Let
$M_{1}\in M_{adj}(m_{1},\mathbb{N})$,
$M_{2}\in M_{adj}(m_{2},\mathbb{N})$,
It is obvious that

$$
\mathbf{diag}(M_{1}, M_{2})\sim
\mathbf{diag}(M_{2}, M_{1}).
$$
Furthermore, we have

$$
\begin{array}{c}
\{\mathbf{diag}(\pi_{1}(M_{1}),\pi_{2}(M_{2}))|
\pi_{i}\in\mathbb{S}_{m_{i}},i=1,2\} \\
\subset
\{\pi(\mathbf{diag}(M_{1},M_{2}))|\pi\in
\mathbb{S}_{m_{1}+m_{2}}\}.
\end{array}
$$
Therefore, we can define

\begin{equation}
\{M_{1}\}\oplus\{M_{2}\}=\{M_{1}\oplus M_{2}\}.
\end{equation}
Based on the previous discussion,
we have

$$
\{M_{1}\}\oplus\{M_{2}\}=
\{M_{2}\}\oplus\{M_{1}\}.
$$
Moreover, it is easy to check that
for $M_{i}\in M_{adj}(m_{i},\mathbb{N})$,
$i=1,2,3$, we have 

$$
(\{M_{1}\}\oplus \{M_{2}\})\oplus\{M_{3}\}=
\{M_{1}\}\oplus(\{M_{2}\}\oplus\{M_{3}\})=
\{\mathbf{diag}(M_{1},M_{2},M_{3})\}.
$$
Thus the direct sum (2.2) is associative and 
commutative. On the other hand, it is obvious that

$$
\mathbf{deg}\{M_{1}+M_{2}\}=
\mathbf{deg}\{M_{1}\}+\mathbf{deg}\{M_{2}\}.
$$

Similarly, in the situation of the extended
adjacency matrices, we take

\begin{equation}
M_{adj}(+\infty,\mathbb{N})_{(e)}=
(\bigcup\limits_{m\geq 2}(M_{adj}(m,\mathbb{N})_{(e)}
\diagup\sim)\setminus\{0\})\cup\{0\}.
\end{equation}

Let $a=(a_{1},\cdots,a_{m})\in\mathbb{N}^{m}$,
$b=(b_{1},\cdots,b_{n})\in\mathbb{N}^{n}$ be 
two multiple indices, we define the direct
sum of $a$ and $b$ denoted by $a\boxplus b$
to be a multiple index in $\mathbb{N}^{m+n}$,

\begin{equation}
a\boxplus b=(a_{1},\cdots,a_{m},
b_{1},\cdots,b_{n})\in\mathbb{N}^{m+n}.
\end{equation}
Especially, let $k,l\in\mathbb{N}$, we define
$k\boxplus a=(k,a_{1},\cdots,a_{m})$ and
$k\boxplus l=(k,l)$.

For two extended adjacency matrices
$(M_{i},b_{i})$ ($M_{i}\in 
M_{adj}(m_{i},\mathbb{N}),\,
b_{i}\in\mathbb{N}^{m_{i}}$), we define
their direct sum in the following way:

\begin{equation}
(M_{1},b_{1})\oplus(M_{2},b_{2})=
(M_{1}\oplus M_{2},b_{1}\boxplus b_{2}).
\end{equation} 
$(M_{1},b_{1})\oplus(M_{2},b_{2})$ is also
expressed by a block matrix as following,

$$
(M_{1},b_{1})\oplus(M_{2},b_{2})=
\begin{pmatrix}
M_{1} & 0 & b_{1}^{T} \\
0 & M_{2} & b_{2}^{T} \\
b_{1} & b_{2} & 0
\end{pmatrix}.
$$
It is obvious that

$$
\{\pi_{1}((M_{1},b_{1}))\oplus
\pi_{2}((M_{2},b_{2}))|\pi_{i}
\in\mathbb{S}_{m_{i}},i=1,2\}\subset
\{(M_{1},b_{1})\oplus((M_{2},b_{2}))\},
$$
thus, we do not need to distingush 
$\{(M_{1},b_{1})\}\oplus\{(M_{2},b_{2})\}$
and $\{(M_{1},b_{1})\oplus(M_{2},b_{2})\}$.
In the other words, we have

$$
\{(M_{1},b_{1})\}\oplus\{(M_{2},b_{2})\}
=\{(M_{1},b_{1})\oplus(M_{2},b_{2})\}.
$$
In the situation of the equivalent classes,
the direct sum is commutative, i.e.
we have

$$
\{(M_{1},b_{1})\}\oplus\{(M_{2},b_{2})\}=
\{(M_{2},b_{2})\}\oplus\{(M_{1},b_{1})\}.
$$

\begin{definition}
Let $\{M\}\in M_{adj}(m,\mathbb{N})\diagup\sim$.
\begin{itemize}
\item When $m\geq 4$, if there are 
$M_{1}\in M_{adj}(k,\mathbb{N})$,
$M_{2}\in M_{adj}(m-k,\mathbb{N})$, such that

$$
\{M\}=\{M_{1}\}\oplus\{M_{2}\},
$$
where $M_{1},M_{2}\not=0$, $k\geq 2$, $m-k\geq 2$,
we say $\{M\}$ is disconnected.
Otherwise, we say $\{M\}$ is connected.
\item When $2\leq m\leq 3$,
if on each row (colunm) of $M$ 
there is a non-zero entry, we say $\{M\}$ is connected.
\item An adjacency matrix $M\in M_{adj}(m,\mathbb{N})$
is called a connected one if $\{M\}$ is connected.
\end{itemize}
\end{definition}

\begin{remark}
\begin{itemize}

\item It is well known that the adjacency 
matrices arises from graphic theory to 
characterize the graphs. In other words, the
adjacency matrices can be regarded as "coordinates"
of the graphs, and the graphs indicate the geometric
meaning of the adjacency matrices. The condition
of zero diagonal indicates the graphs without
loops. The connectivity of the adjacency matrices
defined in definition 2.2 is equivalent to 
the connectivity of the graphs.

\item Let $M\in M_{adj}(2,\mathbb{N})$, then

$$
M\,\,is\,\,connected\,\Leftrightarrow\,M\not=0.
$$

\item We say an extended adjacency matrix
$(M,a)$ is connected, if $M$ is connected.

\item We define the zero matrix is connected.
\end{itemize}
\end{remark}

\begin{proposition}
Let $\{M\}\in M_{adj}(m,\mathbb{N})\diagup\sim$ be  
disconnected, then $\{M\}$ adapts
the decomposition as follows

\begin{equation}
\{M\}=\{M_{1}\}\oplus\cdots\oplus\{M_{k}\},
\end{equation}
where each $\{M_{i}\}\in 
M_{adj}(m_{i},\mathbb{N})\diagup\sim$
is connected ($i=1,\cdots,k,\,\,m_{1}+\cdots+m_{k}=m$).
\end{proposition}

\begin{corollary}
Let $M\in M_{adj}(m,\mathbb{N})$, then
$M$ is disconnected if and only if there is
a partition $\{I_{i}\}_{i=1}^{k}\in
\mathbf{Part}([m])$ ($k\geq 2$), such that
each $M_{I_{i}}$ is connected ($i=1,\cdots,k$),
and $M\sim M_{I_{1}}\oplus\cdots\oplus M_{I_{k}}$.
\end{corollary}

\begin{corollary}
Under the direct sum (2.2), $M_{adj}(+\infty,\mathbb{N})$
is a commutative monoid generated by all connected classes. 
\end{corollary}

\begin{corollary}
Let $M\in M_{adj}(m,\mathbb{N})$,
$I,J\subset[m]$ be two subsets
satisfuing:

\begin{itemize}
\item $J\subset I$,
\item $M_{I}\sim M_{I_{1}}\oplus\cdots\oplus M_{I_{k}}$,
where $\{I_{i}\}_{i=1}^{k}\in\mathbf{Part}(I)$, each
$M_{I_{i}}$ is connected ($i=1,\cdots,k$),
\item $M_{J}\sim M_{J_{1}}\oplus\cdots\oplus M_{J_{l}}$,
where $\{J_{j}\}_{j=1}^{l}\in\mathbf{Part}(J)$, each
$M_{J_{j}}$ is connected ($j=1,\cdots,l$).
\end{itemize}
Then, for each $J_{j}$, there a $I_{i}$ such that
$J_{j}\subset I_{i}$ ($1\leq i\leq k,\,1\leq j\leq l$).
\end{corollary}

All discussions about conncetedness
can be generalized to the situation of the
extended adjacency matrices.

\subsection{Quotient}

Let $M\in M_{adj}(m,\mathbb{N})$,
$I=\{i_{1},\cdots,i_{k}\}\subset[m]$
($k\geq 2,\,0<i_{1}<\cdots<i_{k}$), 
then $I$ determines a diagonal submatrix 
$M_{I}=(m_{i_{a},i_{b}})_{k\times k}$ 
of $M$. In fact the subset $I$ determines
a homomorphism of the monoids
 
$$
\mathfrak{R}_{I}:M_{adj}(m,\mathbb{N})
\longrightarrow M_{adj}(k,\mathbb{N}),\,
\mathfrak{R}_{I}:M\mapsto M_{I}.
$$
Conversely, for the given subset
$I\subset[m]$ as above, we 
can define an embedding $\iota_{I}:M_{adj}(k,\mathbb{N})
\hookrightarrow M_{adj}(m,\mathbb{N})$
in the following way. Let $N=(n_{ij})_
{k\times k}\in M_{adj}(k,\mathbb{N})$,
then $\iota_{I}N\in M_{adj}(m,\mathbb{N})$
with the form $\iota_{I}N=(m^{\prime}_{ij})_{m\times m}$
satisfying 
$m^{\prime}_{i_{i}i_{j}}=n_{ij}$ ($i,j=1,\cdots,k$),
$m^{\prime}_{pq}=0$ 
($p\in I^{c}$ or $q\in I^{c}$, $I^{c}=
[m]\setminus I$).
It is obvious that
$(\iota_{I}M_{I})_{I}=M_{I}$.
For another subset 
$J\subset[m]$, if $J\subset I$,
then $M_{J}=(M_{I})_{J}$. 

We now define the quotient of
$M$ by $M_{I}$ as follows.

\begin{definition}
Let $m\geq 2$ be an integer,
$I=\{i_{1},\cdots,i_{k}\}\subset[m]$,	
$I^{c}=[m]\setminus I
=\{j_{1},\cdots,j_{m-k}\}$
($m\geq k\geq 2,\,0<i_{1}<\cdots<i_{k},
\,j_{1}<\cdots<j_{m-k}$).

\begin{itemize}
\item The quotient is a map 

$$
\mathcal{Q}_{m,I}:M_{adj}(m,\mathbb{N})
\longrightarrow
M_{adj}(m-k+1,\mathbb{N}).
$$
For $M=(m_{ij})_{m\times m}
\in M_{adj}(m,\mathbb{N})$,
$\mathcal{Q}_{m,I}(M)$ is called
the quotient of $M$ by $M_{I}$  
defined by the following expression:

\begin{equation}
\mathcal{Q}_{m,I}(M)=
\begin{pmatrix}
0 & m_{1^{\ast},j_{1}}\cdots m_{1^{\ast},j_{m-k}} \\
\begin{array}{c}
m_{j_{1},1^{\ast}} \\ \vdots \\ m_{j_{m-k},1^{\ast}}
\end{array}
& M_{I^{c}}
\end{pmatrix},
\end{equation}
where $m_{j_{b},1^{\ast}}=m_{1^{\ast},j_{b}}
=\sum_{a=1}^{k}m_{i_{a},j_{b}}$ ($b=1,\cdots,m-k$).
The matrix (2.4) is also denoted by
$M\diagup M_{I}$. We define $M\diagup M=0$,
$M\diagup 0=M$.

\item Let $(M,b)\in M_{adj}(m+1,\mathbb{N})_{(e)}$, 
the quotient of $(M,b)$ by $M_{I}$  
is defined to be an extended adjacency matrix in 
$M_{adj}(m-k+2,\mathbb{N})_{(e)}$, denoted by 
$(M,b)\diagup M_{I}$, with the following form,
	
\begin{equation}
(M,b)\diagup M_{I}=
(M\diagup N_{I},b_{\ast}\boxplus b_{I^{c}}),
\end{equation}
where $b_{\ast}=\sum_{i\in I}b_{i}$,
$b_{I^{c}}=(b_{j_{1}},\cdots,b_{j_{m-k}})$.
\end{itemize}
\end{definition}

There is a basic property as follows.

\begin{lemma}
Let $M\in M_{adj}(m,\mathbb{N})$,
$I\subset K\subset[m]$.
Then we have

$$
M_{K}\diagup M_{I}=(M\diagup M_{I})_{J},
$$
where $J=(K\setminus I)\cup\{1^{\ast}\}$.
\end{lemma}

\begin{proof}
Without loss of the generality, we assume
$I=[n]$, $K=[n+r]$ ($n+r<m$). Thus we have
$J=\{1^{\ast}\}\cup\{n+1,\cdots,n+r\}$. Let

$$
M\diagup N_{I}=
\begin{pmatrix}
m_{1^{\ast}1^{\ast}} & m_{1^{\ast}2}
\cdots m_{1^{\ast}r+1} & m_{1^{\ast}r+2}
\cdots m_{1^{\ast}m-n+1} \\
\begin{array}{c}
m_{21^{\ast}} \\
\vdots \\
m_{r+11^{\ast}}
\end{array}
& M_{K\setminus I} & M_{12}^{T} \\
\begin{array}{c}
m_{r+21^{\ast}} \\
\vdots \\
m_{m-n+11^{\ast}}
\end{array}
& M_{12} & M_{K^{c}}
\end{pmatrix},
$$
then 

$$
(M\diagup N_{I})_{J}=
\begin{pmatrix}
m_{1^{\ast}1^{\ast}} & m_{1^{\ast}2}
\cdots m_{1^{\ast}r+1} \\
\begin{array}{c}
m_{21^{\ast}} \\
\vdots \\
m_{r+11^{\ast}}
\end{array}
& M_{K\setminus I} 
\end{pmatrix}.
$$
By a straightforward calculation we
can get the conclusion of lemma.
\end{proof}

\begin{remark}
\begin{itemize}
	
\item In the present article, we focus on the 
adjacency matrices with zero diagonal which 
correspond to the graphs without loops. 
In the general situation, the entrices on 
diagonal may be non-zero. The quotient in
definition 2.3 can be generalized to the situation
of the adjacency matrices with non-zero
diagonal. For instance, from the geometric 
viewpoint, we 
consider the quotient of a Feynman diagram
by a subdiagram. Recalling that a subdiagram of 
a Feynman diagram is determined by a subset
of the internal lines, or spanned by a
subset of the internal lines, and a
subgraph is spanned by a subset of the vertices,
therefore, the adjacency matrix
characterizing this quotient should be
of the form

$$
\mathbf{diag}(m_{1^{\ast}1^{\ast}},0,
\cdots,0)+M\diagup M_{I},
$$
where $M_{I}$ indeecates the subgraph with
same vertices as subdiagram metioned above,
$m_{1^{\ast}1^{\ast}}$ indecates the number
of the loops arising from the procedure
of the quotient ($0\leq m_{1^{\ast}1^{\ast}}
<\mathbf{deg}M_{I}$). 
The above adjacency matrix shows that
when we discuss the quotient, it is enough
for us to consider the situation of the
graphs without loops. The situation of the
extended adjacency matrices is similar.

\item Let $M,\in M_{adj}(m,\mathbb{N})$
(or $(M,a)\in M_{adj}(m,\mathbb{N})_{(e)}$),
$\{I_{i}\}_{i=1}^{k}$ be a sequence 
consisting of disjoint subsets of $[m]$,
i.e. $I_{i}\cap I_{i^{\prime}}=\emptyset$
($i\not=i^{\prime},\,1\leq i,i^{\prime}\leq k$).
We can make quotient repeatedly as follows,

$$
(\cdots((M\diagup M_{I_{1}})\diagup 
M_{I_{2}})\cdots)\diagup M_{I_{k}}.
$$
or

$$
(\cdots(((M,a)\diagup M_{I_{1}})\diagup 
M_{I_{2}})\cdots)\diagup M_{I_{k}}.
$$
We denote above quotient by $M\diagup(M_{I_{i}})$
($(M,a)\diagup(M_{I_{i}})$)for short. 
If $|I_{1}|+\cdots+|I_{k}|=n$, then 
$M\diagup(M_{I_{i}})\in M_{adj}(m-n+k,\mathbb{N})$.
Precisely, we have

$$
\begin{array}{c}
(\cdots((M\diagup M_{I_{1}})\diagup 
M_{I_{2}})\cdots)\diagup M_{I_{k}} \\
=
\begin{pmatrix}
\begin{array}{ccccccc}
0 & m_{1^{\ast}2^{\ast}} &\cdots & &
&\cdots & m_{1^{\ast}\,m-n+k} \\
m_{2^{\ast}\,1^{\ast}} &\ddots &\cdots & & &
\cdots & m_{2^{\ast}\,m-n+k}  \\
\vdots &\vdots &\ddots & & & &\vdots \\
& & & 0 
& m_{k^{\ast}k+1} &\cdots &
m_{k^{\ast}m-n+k} \\
& & & m_{k+1\,k^{\ast}} & & & \\
\vdots & & &\vdots & & M_{I^{c}} & \\
m_{m-n+k\,1^{\ast}} &\cdots & & 
m_{m-n+k\,k^{\ast}} & & & 
\end{array}
\end{pmatrix},
\end{array}
$$
where $I=\bigcup_{i=1}^{k}I_{i}$  and 
$1th,\cdots,ith$ rows (or
$1th,\cdots,ith$ columns) in $M\diagup(M_{I_{i}})$
consist of ideal entries arising from the
quotient. It is obvious that we have
$M\diagup(M_{I_{i}})\sim M\diagup(M_{I_{\sigma(i)}})$
for each $\sigma\in\mathbb{S}_{k}$.
The situation of the extended adjacency matrices 
is similar.
\end{itemize}

\end{remark}

The notation of the quotient can be extended
into the situation of equivalent classes.
Actually, we have the following lemma.

\begin{proposition}
Let $M\in M_{adj}(m,\mathbb{N})$, 
$I\subset[m]$. 
Then, each $\pi\in\mathbf{S}_{m}$ induces
a permutation $\pi_{I}\in\mathbf{S}_{|I|}$ 
such that
	
$$
\pi_{I}(M_{I})=\pi(M)_{\pi^{-1}(I)},
$$
and	
	
\begin{equation}
M\diagup M_{I}\sim\pi(M)
\diagup\pi_{I}(M_{I}).
\end{equation}
\end{proposition}

\begin{proof}
Let $|I|=k$, $I=\{i_{1},\cdots,i_{k}\}
\subset[m]$,	
$I^{c}=[m]\setminus I
=\{j_{1},\cdots,j_{m-k}\}$
($m\geq k\geq 2,\,0<i_{1}<\cdots<i_{k}$,
$0<j_{1}<\cdots<j_{m-k}$), 
$\pi\in\mathbf{S}_{m}$.
Then $\pi$ induces a permutation $\pi_{I}\in\mathbf{S}_{k}$
acting on $I$. Actually, let $\pi^{-1}(I)=
\{l_{1},\cdots,l_{k}\}$, where $\pi^{-1}(i_{a})=l_{a}$
($a=1,\cdots,k$), then $\pi(M)_{\pi^{-1}(I)}=
(m_{\pi(l_{\alpha_{a}})\pi(l_{\alpha_{b}})})_{k\times k}$,
where $0<l_{\alpha_{1}}<\cdots<l_{\alpha_{k}}$.
Thus, we get a permutation $\pi_{I}\in\mathbf{S}_{k}$,

$$
\pi_{I}=
\begin{pmatrix}
1 & \cdots & k \\
\alpha_{1} & \cdots & \alpha_{k}
\end{pmatrix}.
$$
Let $\pi_{I}$ acts on $I$ in such a way
$\pi_{I}(i_{a})=i_{\pi_{I}(a)}=i_{\alpha_{a}}$ 
($a=1,\cdots,k$),
then, it is obvious that we have

$$
\pi(M)_{\pi^{-1}(I)}=\pi_{I}(M_{I}).
$$
Similarly, $\pi$ induces a permutation 
$\pi_{I^{c}}\in\mathbf{S}_{m-k}$ such that

$$
\pi_{I^{c}}(M_{I^{c}})=\pi(M)_{\pi^{-1}(I^{c})}.
$$
By a straightforward calculation, 
we can get
	
$$
\pi(M)\diagup\pi_{I}(M_{I})=
\begin{pmatrix}
 0 & m_{1^{\ast},j_{\tau(1)}},\cdots,
m_{1^{\ast},j_{\tau(m-k)}} \\
\begin{array}{c}
m_{j_{\tau(1)},1^{\ast}} \\ \vdots \\
m_{j_{\tau(m-k)},1^{\ast}}
\end{array}
& \tau(M_{I^{c}})
\end{pmatrix},
$$
where $\tau=\pi_{I^{c}}$.
Let $P_{\tau}$ be a permutation matrix
of order $m-k$ corresponding to the permutation
$\tau\in\mathbf{S}_{m-k}$,
it is obvious that 

$$
\pi(M)\diagup \pi_{I}(M_{I})
=\mathbf{diag}(1,P_{\tau})
(M\diagup N_{I})\mathbf{diag}(1,P_{\tau}^{T}).
$$
	
\end{proof}

\begin{corollary}
Let $(M,a)\in M_{adj}(m+1,\mathbb{N})_{(e)}$, 
$I\subset\underline{m}$, 
$\pi\in\mathbb{S}_{m}$.
Then, there is $\pi_{I}\in\mathbf{S}_{|I|}$
such that
	
$$
(\pi_{I}(M_{I}),\pi_{I}(a_{I}))=
(\pi(M)_{\pi^{-1}(I)},\pi(a)_{\pi^{-1}(I)}),
$$
and
	
$$
(M,a)\diagup M_{I}\sim
(\pi(M),\pi(a))\diagup\pi_{I}(M_{I}).
$$
\end{corollary}

\begin{remark}
\begin{itemize}
\item  Let $m$ be a positive integer, 
$I\subset[m]$, $M\in M_{adj}(m,\mathbb{N})$, 
we call

$$
\{\pi(M)_{\pi^{-1}(I)}|\,\pi\in\mathbb{S}_{m}\}
$$
the diagonal sub-class of $\{M\}$ corresponding 
to $I$, denoted by $\{M\}_{I}$.
By proposition 2.3, we know that for 
$M\in M_{adj}(m,\mathbb{N})$, 
$I\subset[m]$ with 
$|I|=k\,(k\geq 2)$ the quotient 
$M\diagup M_{I}$ defines a map

$$
\{\pi(M)|\,\pi\in\mathbb{S}_{m}\}\mapsto
\{\pi(M)\diagup\pi(M)_{\pi^{-1}(I)}|\,\pi\in
\mathbb{S}_{m}\}\subset\{M\diagup M_{I}\},
$$
thus, a map		
$$
\{M\}\mapsto\{M\diagup M_{I}\}.
$$

\item From definition of the quotient,
it is easy to see
	
$$
\mathbf{deg}\{M\}=\mathbf{deg}\{N_{I}\}+
\mathbf{deg}\{M\diagup N_{I}\}.
$$  

\end{itemize}	
\end{remark}

In definition 2.3 we do not require $M$ and $M_{I}$
are connected. From now on, when we discuss
the quotient $M\diagup M_{I}$ given by the 
expression (2.8), we assume both $M$ and $M_{I}$
are connected. In the situation of that
$M$ is connected and $M_{I}$ is disconnected, 
$M_{I}$ will adapt a decomposition
$M_{I}\sim M_{I_{1}}\oplus\cdots\oplus M_{I_{k}}$,
$\{I_{i}\}_{I=1}^{k}\in\mathbf{Part}(I)$,
each $M_{I_{i}}$ is connnected ($i=1,\cdots,k$).
In this situation, the quotient $M\diagup M_{I}$
will be regarded as

\begin{equation}
M\diagup M_{I}=M\diagup(M_{I_{i}}).
\end{equation}
If $M$ is disconnected, but $M_{I}$ is connected,
$M$ will adapt a decomposition $M\sim M_{J_{1}}
\oplus\cdots\oplus M_{J_{l}}$, where
$\{J_{j}\}\in\mathbf{Part}([m])$, and
each $M_{J_{j}}$ is connected ($j=1,\cdots,l$).
In this situation, there is some $J_{j^{\prime}}$
such that $I\subset J_{j^{\prime}}$. The quotient
$M\diagup N_{I}$ should satisfy

\begin{equation}
M\diagup M_{I}\sim(M_{J_{j^{\prime}}}\diagup M_{I})\oplus
(\bigoplus_{j\not=j^{\prime}}M_{J_{j}}).
\end{equation}
The situation of the extended adjacency matrices
is similar. 

We now give a explicit description
about the quotient.

\begin{proposition}
Let $M\in M_{adj}(m,\mathbb{N})$, 
	
$$
M\sim M_{I_{1}}\oplus\cdots\oplus M_{I_{k}},
$$
where each $M_{I_{i}}$ is connected ($i=1,\cdots,k$),
$\{I_{i}\}_{i=1}^{k}\in\mathbf{Part}(\underline{m})$.
For a subset $J\subset\underline{m}$, let
$M_{J}\sim M_{J_{1}}\oplus\cdots\oplus M_{J_{l}}$,
$\{J_{j}\}_{j=1}^{l}\in\mathbf{Part}(J)$, and each
$M_{J_{j}}$ be connected ($j=1,\cdots,l$). 
Then, the quotient of $M$ by $M_{J}$ 
is of the following form,
	
\begin{equation}
M\diagup M_{J}\sim(\bigoplus
\limits_{I_{i}\cap J=\emptyset}M_{I_{i}})
\oplus(\bigoplus\limits_{I_{i}\cap J\not=\emptyset}
M_{I_{i}}\diagup(M_{J_{j}})_{J_{j}\subset I_{i}}).
\end{equation}
\end{proposition}

\begin{proof}
At first, we consider $M\diagup N_{J_{1}}$. In this 
situation, by corollary 2.4, we know that 
there is some $I_{i^{\prime}}$ such that
$J_{1}\subset I_{i^{\prime}}$. From definition 2.3
we know that
	
$$
M\diagup M_{J_{1}}\sim(\bigoplus\limits_
{i\not=i^{\prime}}M_{I_{i}})\oplus
(M_{I_{i^{\prime}}}\diagup M_{J_{1}}).
$$
Because $M_{I_{i^{\prime}}}\diagup M_{J_{1}}$
is conncected, for $J_{2}$ there are two possibilities.

\begin{itemize}
\item There is some $I_{i^{\prime\prime}}$ such that
$J_{2}\subset I_{i^{\prime\prime}}$ 
($i^{\prime\prime}\not=i^{\prime}$).
Then we have
		
$$
M\diagup(M_{J_{j}})_{j=1,2}\sim(\bigoplus\limits_
{i\not=i^{\prime},i^{\prime\prime}}M_{I_{i}})\oplus
(M_{I_{i^{\prime}}}\diagup M_{J_{1}})
\oplus(M_{I_{i^{\prime\prime}}}\diagup M_{J_{2}}).
$$
		
\item $J_{2}\subset I_{i^{\prime}}\setminus J_{1}$,
then we hvae
		
$$
M\diagup(M_{J_{j}})_{j=1,2}\sim(\bigoplus\limits_
{i\not=i^{\prime}}M_{I_{i}})\oplus
(M_{I_{i^{\prime}}}\diagup(M_{J_{j}})_{j=1,2}).
$$
\end{itemize}
Repeating above procedure, inductively,
we can prove the formula (2.14).
\end{proof}

\begin{proposition}
Let $M\in M_{adj}(m,\mathbb{N})$, 
$I\subset[m]$. 
Then $M\diagup M_{I}$ is connected if
and only if $M$ is connected.
\end{proposition}

\begin{proof}
Without loss of generality, we assume
$M_{I}$ is connected. First we assume
$M$ is connected, but $M\diagup M_{I}$ 
is disconnected, then there is a subset 
$J\subset(\{1^{\ast}\}\cup\{1,\cdots,m-k\})$, 
such that $1^{\ast}\in J$,
$(M\diagup M_{I})_{J}$ is conncted, and
$M\diagup M_{I}\sim(M\diagup M_{I})_{J}
\oplus(M\diagup M_{I})_{J^{c}}$, where
$J^{c}=\{1,\cdots,m-k\}\setminus J$, $k=|I|$. It is 
obvious that $J^{c}\subset I^{c}$, thus
$(M\diagup M_{I})_{J^{c}}=M_{J^{c}}$.
Due to definition 2.3, it is easy to check that
$M\sim M_{J^{\prime}\cup I}\oplus M_{J^{c}}$,
where $J^{\prime}=J\setminus\{1^{\ast}\}$.
Up to now, we reach a contradiction.

Suppose $M\diagup M_{I}$ is connected, 
by proposition 2.4, $M$ is also connected.  
\end{proof}

Regarding the quotient as an operation,
we will prove that the quotient is compatible
with the direct sum of the adjacency matrices.
For two adjacency matrices 
$M\in M_{adj}(m,\mathbb{N})$ and
$N\in M_{adj}(n,\mathbb{N})$, we can identify
$M\oplus N$ with $\mathbf{diag}(M,N)$,
which means we embed $[n]$ into $[m+n]$.
In this situation we will identify
$[n]$ with $\{m+1,\cdots,m+n\}$.
Thus, for a subset $J\subset[n]$,
we do not distinguish between $J$ and
$\{j+m|j\in J\}\subset\{m+1,\cdots,m+n\}$.
Conversely, for any subset $K\subset[m+n]$,
we have a decomposition $K=K_{1}\cup K_{2}$,
where $K_{1}=K\cap[m]$ and
$K_{2}=K\cap\{m+1,\cdots,m+n\}$, 
$K_{2}$ can be regarded as a subset
of $[n]$. Let $M$ adapt the decomposition 
$M\sim M_{I_{1}}\oplus\cdots\oplus M_{I_{k}}$, 
$N$ adapt the decompsition 
$N\sim N_{J_{1}}\oplus\cdots\oplus 
N_{J_{l}}$, where 
$\{I_{i}\}\in\mathbf{Part}(I)$,
$\{J_{j}\}\in\mathbf{Part}(J)$,
each $M_{I_{i}}$ and each $N_{J_{j}}$
are connected ($i=1,\cdots,k,\,j=1,\cdots,l$).
Then we have

$$
M\oplus N\sim(\bigoplus\limits_{i=1}^{k}M_{I_{i}})
\oplus(\bigoplus\limits_{j=1}^{l}N_{J_{j}}).
$$

On the other hand,
if $M_{K_{1}}\sim\bigoplus_{\alpha=1}^{p}M_{D_{\alpha}}$, 
and $N_{K_{2}}\sim\bigoplus_{\beta=1}^{q}N_{E_{\beta}}$,
where $\{D_{\alpha}\}\in\mathbf{Part}(K_{1})$,
$\{E_{\beta}\}\in\mathbf{Part}(K_{2})$,
each $M_{D_{\alpha}}$ and each $N_{E_{\beta}}$
are connected, it is obvious that 
$(M\oplus N)_{K}$
adapts the following decomposition,

$$
(M\oplus N)_{K}\sim(\bigoplus\limits
_{\alpha=1}^{p}M_{D_{\alpha}})\oplus
(\bigoplus\limits_{\beta=1}^{q}N_{E_{\beta}}).
$$
Due to proposition 2.4,
it is easy to check that

$$
\begin{array}{c}
(M\oplus N)\diagup(M\oplus N)_{K} \\
\sim(M\oplus N)\diagup((M_{D_{\alpha}})\cup(N_{E_{\beta}})) \\
\sim(M\diagup(M_{D_{\alpha}}))\oplus(N\diagup(N_{E_{\beta}})) \\
=(M\diagup M_{K_{1}})\oplus(N\diagup N_{K_{2}}).
\end{array}
$$
Summarizing the previous discussion, we 
reach the following conclusion.

\begin{lemma}
Let $M\in M_{adj}(m,\mathbb{N})$,
$N\in M_{adj}(n,\mathbb{N})$,
$K\subset[m+n]$. Then, we have

$$
(M\oplus N)\diagup(M_{K_{1}}\oplus N_{K_{2}})\sim
(M\diagup M_{K_{1}})\oplus(N\diagup N_{K_{2}}),
$$
where $K_{1}=K\cap[m]$,
$K_{2}=K\cap\{m+1,\cdots,m+n\}$,
and we identify $K_{2}$ with the set
$\{k-m|k\in K_{2}\}\subset[n]$.
\end{lemma}

Equivalently, we have

\begin{corollary}
Let $M\in M_{adj}(m,\mathbb{N})$,
$N\in M_{adj}(n,\mathbb{N})$,
$K\subset[m]$. Then, we have

$$
\{(M\oplus N)\diagup(M_{K_{1}}\oplus N_{K_{2}})\}=
\{M\diagup M_{K_{1}}\}\oplus\{N\diagup N_{K_{2}}\},
$$
where $K_{1}=K\cap[m]$,
$K_{2}=K\cap\{m+1,\cdots,m+n\}$,
and we identify $K_{2}$ with the set
$\{k-m|k\in K_{2}\}\subset[n]$.
\end{corollary}

Now we turn to more complicated situation
of the quotient.

\begin{proposition}
Let $m\geq 2$ be an integer, and

\begin{itemize}
\item $\{I_{i}\}_{i=1}^{p}\in\mathbf{part}([m])$,
$\bigcup_{i=1}^{p}I_{i}=I\subset K\subset[m]$,
\item $M\in M_{adj}(m,\mathbb{N})$, 
 
\end{itemize}
If we take $J=(K\setminus I)\cup
\{1^{\ast},\cdots,p^{\ast}\}$,
then we have

\begin{itemize}

\item 

$$
(M\diagup(M_{I_{i}}))_{J}=M_{K}\diagup(M_{I_{i}}),
$$
\item

$$
(M\diagup(M_{I_{i}}))\diagup(M\diagup(M_{I_{i}}))_{J}
=M\diagup M_{K}.
$$
\end{itemize}
Particularly, if each $M_{I_{i}}$ is
connected ($i=1,\cdots,p$), 
then $(M\diagup(M_{I_{i}}))_{J}$ is connected
if and only if $M_{K}$ is connected.	
\end{proposition}

\begin{proof}
We will prove the conclusion by induction on $p$.	
The situation of $p=1$ has been proven in
lemma 2.1. Assuming the conclusion is valid 
for $p>0$, we consider the situation of $p+1$.
Noting that

$$
M\diagup(M_{I_{i}})_{i=1}^{p+1}=
(M\diagup(M_{I_{i}})_{i=1}^{p})
\diagup M_{I_{p+1}}
=\tilde{M}\diagup M_{I_{p+1}},
$$
where $\tilde{M}= M\diagup(M_{I_{i}})_{i=1}^{p}$,
then there is $J^{\prime}=((K\setminus I)\cup 
I_{p+1})\cup\{1^{\ast},\cdots,p^{\ast}\}$,
such that

$$
(\tilde{M}\diagup M_{I_{p+1}})\diagup
(\tilde{M}\diagup M_{I_{p+1}})_{J}=
\tilde{M}\diagup\tilde{M}_{J^{\prime}}.
$$
By the inductive hypothesis, 
we know that thr conclusion of proposition
is valid.
	
\end{proof}

The following conclusion concerns the general
situation.

\begin{proposition}
Let $M\in M_{adj}(m,\mathbb{N})$ ($m\geq 2$), 
$\{I_{i}\}_{i=1}^{p},
\{K_{k}\}_{k=1}^{q}\in\mathbf{part}([m])$,
$\{I_{i}\}\subset\{K_{k}\}$.
If we take $\{J_{j}\}\in\mathbf{part}(([m]\setminus I)
\cup\{1^{\ast},\cdots,p^{\ast}\})$ with the
following form

$$
J_{j}=(K_{k}\setminus(\bigcup_{I_{i}\subset K_{k}}
I_{i}))\cup\{i^{\ast}\}_{I_{i}\subset K_{k}},\,
K_{k}\not=I_{i}\,(1\leq i\leq p),
$$
where $I=\bigcup_{i}I_{i}$,
then we have
	
\begin{itemize}

\item
		
$$
(M\diagup(N_{I_{i}}))_{J_{j}}=M_{K_{k}}\diagup
(I_{i})_{I_{i}\subset K_{k}},
K_{k}\not=I_{i}\,(1\leq i\leq p),
$$

\item
		
$$
(M\diagup(N_{I_{i}}))\diagup
((M\diagup(N_{I_{i}}))_{J_{j}})=
M\diagup(M_{K_{k}}).
$$
\end{itemize}
\end{proposition}

\begin{proof}
	Let $\{I_{i^{\prime}}\}=\{I_{i}\}_
	{I_{i}\cap J^{\prime}=\emptyset}$.	
	Without loss of generality, we can
	replace $M$ by $M\diagup(I_{i^{\prime}})$.	
	Thus, we can assume that each $K_{k}$	
	($1\leq k\leq q$) satisfies the following
	condition:
	
	$$
	|\{I_{i}|I_{i}\subset K_{k}\}|=1\,
	\Rightarrow\,
	K_{k}\setminus I\not=\emptyset.
	$$
	
	We now prove the conclusion for each
	
	$$
	J_{k}=(K_{k}\setminus(\bigcup_{I_{i}\subset K_{k}}
	I_{i}))\cup\{i^{\ast}\}_{I_{i}\subset K_{k}},\,
	k=1,\cdots,p.
	$$
	For instance, we consider $J_{1}$. For
	simplicity, we assume
	
	$$
	J_{1}=J_{1}^{\prime}\cup
	\{1^{\ast},\cdots,r^{\ast}\},\,
	1\leq r\leq p,
	$$
	where $J_{1}^{\prime}=K_{1}\setminus
	(\bigcup_{i=1}^{r}I_{i})$.
	We want to prove
	
	$$
	(M\diagup(N_{I_{i}}))\diagup(M\diagup(N_{I_{i}}))_{J_{1}}=
	(M\diagup(I_{i})_{i>r})\diagup M_{K_{1}},
	$$
	and
	
	$$
	(M\diagup(N_{I_{i}}))_{J_{1}}=
	M_{K_{1}}\diagup(N_{I_{i}})_{1\leq i\leq r}.
	$$
	It is obvious that we have
	
	$$
	M\diagup(I_{i})=(M\diagup(I_{i})_{i>r})
	\diagup(I_{i})_{1\leq i\leq r},
	$$
	and
	
	$$
	(M\diagup(I_{i})_{i>r})_{K_{1}}=M_{K_{1}}.
	$$
	If we take $\tilde{M}=M\diagup(I_{i})_{i>r}$, the 
	situation is reduced to Lemma 2.6.
	
\end{proof}

Conversely, we have,

\begin{proposition}
Let $M\in M_{adj}(m,\mathbb{N})$ ($m\geq 2$), 
$\{I_{i}\}_{i=1}^{p}\in\mathbf{part}([m])$,
$\{J_{j}\}_{j=1}^{q}\in
\mathbf{part}(([m]\setminus I)
\cup\{1^{\ast},\cdots,p^{\ast}\})$,
where $I=\bigcup_{i=1}^{p}I_{i}$.
If we take $\{K_{k}\}\in\mathbf{part}([m])$ 
with the following form

$$
K_{k}=\left\{
\begin{array}{cc}
I_{i}, & I_{i}\cap J^{\prime}=\emptyset, \\
J_{j}, & J_{j}\cap\{1^{\ast},
\cdots,p^{\ast}\}=\emptyset, \\
J_{j}^{\prime}\cup(\bigcup_{i^{\ast}\in J_{j}}I_{i}),
& J_{j}\cap\{1^{\ast},
\cdots,p^{\ast}\}\not=\emptyset,
\end{array}
\right.
$$
where $J^{\prime}=(\bigcup_{j}J_{j})\setminus
\{1^{\ast},\cdots,p^{\ast}\}$,
$J_{j}^{\prime}=J_{j}\cap J^{\prime}$,
then, we have
	
\begin{itemize}
		
\item
		
$$
(M\diagup(N_{I_{i}}))_{J_{j}}=
\left\{
\begin{array}{cc}
M_{K_{k}}\diagup(I_{i})_{I_{i}\subset K_{k}},
& J_{j}\cap\{1^{\ast},\cdots,p^{\ast}\}
\not=\emptyset, \\
M_{J_{j}}, & J_{j}\cap\{1^{\ast},\cdots,p^{\ast}\}
=\emptyset.
\end{array}
\right.
$$
\item
		
$$
(M\diagup(N_{I_{i}}))\diagup
((M\diagup(N_{I_{i}}))_{J_{j}})=
M\diagup(M_{K_{k}}).
$$
\end{itemize}
\end{proposition}

\begin{proof}
	Let $\{J_{j^{\prime}}\}=\{J_{j}\}_{J_{j}
		\cap\{1^{\ast},\cdots,p^{\ast}\}=\emptyset}$, 
	then we can replace $\{I_{i}\}$ by
	$\{I_{i}\}\cup\{J_{j^{\prime}}\}$. Thus,
	we assume
	
	$$
	\{K_{k}\}=\{I_{i}\}_
	{I_{i}\cap J^{\prime}=\emptyset}
	\cup\{J_{j}^{\prime}
	(\bigcup_{i^{\ast}\in J_{j}}I_{i})\}.
	$$
	
	For instance, we consider the situation of 
	$J_{1}$, and for simplicity we assume
	$\bigcup_{i^{\ast}\in J_{1}}I_{i}=
	\{1,\cdots,r\}$ and $K_{1}=J_{1}^{\prime}
	\cup\{1,\cdots,r\}$. By the same reason as
	proposition 2.6, we know that
	$(M\diagup(I_{i}))_{J_{1}}=M_{K_{1}}
	\diagup(I_{i})_{1\leq i\leq r}$. With the
	arguments which are similar to the ones in
	the the proofs of proposition 2.6 and proposition
	2.7, we can prove
	
	$$
	(M\diagup(M_{I_{i}}))\diagup((M\diagup(M_{I_{i}}))_{J_{j}})=
	M\diagup(M_{K_{k}}).
	$$
	
\end{proof}

\begin{corollary}
Let $M\in M_{adj}(m,\mathbb{N})$ be a 
connected adjacency matrices. For two
subsets $I,K\subset\underline{m}$, If 
	
\begin{itemize}
\item $I\subset K$, 
\item $M_{K}\sim M_{K_{1}}\oplus\cdots
\oplus M_{K_{p}}$, where $\{K_{k}\}\in\mathbf{Part}(K)$,
and each $M_{K_{k}}$ is connected ($k=1,\cdots,p$),
\item $M_{I}\sim M_{I_{1}}\oplus\cdots
\oplus M_{I_{q}}$, where $\{I_{i}\}\in\mathbf{Part}(I)$,
each $M_{I_{i}}$ is connected ($i=1,\cdots,q$),
\end{itemize}
then, there is a subset
$J\subset(K\setminus I)\cup
\{1^{\ast},\cdots,q^{\ast}\}$ such that

\begin{itemize}
\item 

$$
M_{K}\diagup M_{I}=(M\diagup M_{I})_{J},
$$
\item

\begin{equation}
(M\diagup M_{I})\diagup(M\diagup N_{I})_{J}
=M\diagup M_{K}.
\end{equation}
\item $(M\diagup M_{I})_{J}\sim
(M\diagup M_{I})_{J_{1}}\oplus\cdots\oplus
(M\diagup M_{I})_{J_{l}}$, where
$\{J_{j}\}$ is same as one in proposition
2.7, and each $(M\diagup M_{I})_{J_{j}}$ is
connected ($j=1,\cdots,l$).
\end{itemize}	
\end{corollary}

Conversely, we have the following conclusion.

\begin{corollary}
Let $M,\in M_{adj}(m,\mathbb{N})$,
be a connected adjacency matrix,
$I\subset[m]$, 
$M_{I}\sim M_{I_{1}}\oplus\cdots\oplus M_{I_{p}}$,
$\{I_{i}\}_{i=1}^{p}\in\mathbf{Part}(I)$,
each $M_{I_{i}}$ is connected ($i=1,\cdots,p$),
Then, for each subset
$J\subset([m]\setminus I)\cup
\{1^{\ast},\cdots,p^{\ast}\}$,
there is a subset $K\subset[m]$ 
satisfying the following conditions:
	
\begin{itemize}
\item $K=I\cup(J\cap[m])$,
		
\item 
		
\begin{equation}
(M\diagup M_{I})\diagup(M\diagup M_{I})_{J}
=M\diagup M_{K}.
\end{equation}
\end{itemize}
More precisely, if 
	
$$
(M\diagup M_{I})_{J}\sim
(M\diagup M_{I})_{J_{1}}\oplus\cdots\oplus 
(M\diagup M_{I})_{J_{q}},
$$
where $\{J_{j}\}_{j=1}^{q}\in\mathbf{Part}(J)$ 
and each $(M\diagup M_{I})_{J_{j}}$ is connected
($j=1,\cdots,q$), then 
	
$$
M_{K}\sim M_{K_{1}}\oplus\cdots\oplus M_{K_{r}},
$$
where $\{K_{k}\}\in\mathbf{Part}(K)$ and each $M_{K_{k}}$
is connected ($k=1,\cdots,r$), additionally, 
$\{K_{k}\}$ and $M_{K_{k}}$ satisfy:
	
\begin{itemize}
\item 
		
$$
\{K_{k}\}=\{I_{i}\}_{i^{\ast}\notin J}
\cup\{J_{j}\}_{J_{j}\cap\{1^{\ast},\cdots,l^{\ast}\}=\emptyset}
\cup\{L_{J_{j}}\}_{J_{j}\cap
\{1^{\ast},\cdots,l^{\ast}\}\not=\emptyset},  
$$
where $L_{J_{j}}=J_{j}^{\prime}\cup
(\bigcup_{i^{\ast}\in J_{j}}I_{i})$,
$J_{j}^{\prime}=J_{j}\cap[m]$.
		
\item

$$
(M\diagup M_{I})_{J_{j}}=\left\{
\begin{array}{cc}
M_{L_{J_{j}}}\diagup(I_{i})_{i^{\ast}\in J_{j}},
 & J_{j}\cap\{1^{\ast},\cdots,l^{\ast}\}
\not=\emptyset, \\
M_{J_{j}}, & J_{j}\cap\{1^{\ast},\cdots,l^{\ast}\}
=\emptyset.
\end{array}\right.
$$
 
\end{itemize}
\end{corollary}

\subsection{The coproduct}

Let

\begin{equation}
\mathcal{H}_{adj}=\mathbf{Span}_{\mathbb{C}}
(M_{adj}(+\infty,\mathbb{N})).
\end{equation}
The direct sum in $M_{adj}(+\infty,\mathbb{N})$
can be extended to the situation of tensor, thus, the tensor
of $\mathcal{H}_{adj}$. Let $\{M_{i}\},\{N_{i}\}\in
M_{adj}(+\infty,\mathbb{N})$ ($i=1,2$), 
it is natural for us to define the direct sum
of the tensor in the following way.

$$
(\{M_{1}\}\otimes\{M_{2}\})\oplus
(\{N_{1}\}\otimes\{N_{2}\})=
(\{M_{1}\}\oplus\{N_{1}\})\otimes
(\{M_{2}\}\oplus\{N_{2}\}).
$$
Above direct sum is obviously well defined
and can be extened to the situation of the
tensor with multiple factors.

we now define the coproduct on $\mathcal{H}_{adj}$.

\begin{definition}
The coproduct on $\mathcal{H}_{adj}$ is
defined as follows.	
	
\begin{itemize}
\item Let $\{M\}\in M_{adj}(m,\mathbb{N})\diagup\sim$
be connected, $\{M\}\not=0$, we define the 
coproduct as follows,

\begin{equation}
\bigtriangleup\{M\}=\{M\}\otimes 
\{0\}+\{0\}\otimes\{M\}+
\sum\limits_{I\subset[m],\,I\not=[m]}
\{M_{I}\}\otimes\{M\diagup M_{I}\}.
\end{equation}

\item 

$$
\bigtriangleup\{0\}=
\{0\}\otimes\{0\}.
$$
\item Let $\{M_{i}\}\in M_{adj}(m_{i},\mathbb{N})$,
each $\{M_{i}\}$ be connected ($i=1,\cdots,n$).
Then we define

$$
\bigtriangleup(\bigoplus_{i=1}^{n}\{M_{i}\})
=\bigoplus_{i=1}^{n}\bigtriangleup\{M_{i}\}.
$$
\end{itemize}	
\end{definition}

We now prove the co-associativity of the
coproduct $\bigtriangleup$.

\begin{theorem}
The coproduct $\bigtriangleup$ satisfies
the following formula,

$$
(1\otimes\bigtriangleup)\bigtriangleup=
(\bigtriangleup\otimes 1)\bigtriangleup.
$$
\end{theorem}

\begin{proof}
Let $\{M\}\in M_{adj}(m,\mathbb{N})\diagup\sim$
be connected, $\{M\}\not=0$, we first consider 
the left side of the formula in theeorem 2.1. It is

$$
\begin{array}{c}
(1\otimes\bigtriangleup)\bigtriangleup\{M\} \\
=(1\otimes\bigtriangleup)(\{M\}\otimes\{0\}+
\{0\}\otimes\{M\}+\sum\limits_{I\subset[m],\,I\not=[m]}
\{M_{I}\}\otimes\{M\diagup M_{I}\})  \\
=\{M\}\otimes\{0\}\otimes\{0\}+
\{0\}\otimes\bigtriangleup\{M\}+
\sum\limits_{I\subset[m],\,I\not=[m]}
\{M_{I}\}\otimes
\bigtriangleup\{M\diagup M_{I}\},
\end{array} 
$$
where

$$
\begin{array}{c}
\bigtriangleup\{M\diagup M_{I}\}=
\{M\diagup M_{I}\}\otimes\{0\}+
\{0\}\otimes\{M\diagup M_{I}\}+ \\
\sum\limits_{J\subset
([m]\setminus I)\cup\{1^{\ast},\cdots,p^{\ast}\}}
\{(M\diagup M_{I})_{J}\}
\otimes\{(M\diagup M_{I})\diagup(M\diagup M_{I})_{J}\}.
\end{array}
$$
In above sum, the positive integer $p$ arises 
from the decomposition of $M_{I}$,
that is $M_{I}\sim M_{I_{1}}\oplus\cdots\oplus M_{I_{p}}$,
where each $M_{I_{i}}$ is connected ($i=1,\cdots,p$).
According to corollary 2.6 and corollary
2.7, we know that for each 
$J\subset([m]\setminus I)\cup
\{1^{\ast},\cdots,p^{\ast}\})$
there is $K\subset[m]$
such that $(M\diagup M_{I})_{J}=M_{K}\diagup M_{I}$,
$(M\diagup M_{I})\diagup
(M\diagup M_{I})_{J}=M\diagup M_{K}$,
and vice versa. In summary, we have

$$
\begin{array}{c}
(1\otimes\bigtriangleup)\bigtriangleup\{M\} 
=\{M\}\otimes\{0\}\otimes\{0\}+
\{0\}\otimes\bigtriangleup\{M\}+ \\
\sum\limits_{I\subset[m],\,I\not=[m]}
\{M_{I}\}\otimes
(\{M\diagup M_{I}\}\otimes\{0\}+
\{0\}\otimes\{M\diagup M_{I}\})+ \\
\sum\limits_{I\subset K\subset[m],\,K\not=[m]}
\{M_{I}\}\otimes\{M_{K}\diagup M_{I}\}
\otimes\{M\diagup M_{K}\}.
\end{array}
$$

We now consider the right side of the formula
in theorem 2.1. We have

$$
\begin{array}{c}
(\bigtriangleup\otimes 1)
\bigtriangleup\{M\} \\
=(\bigtriangleup\otimes 1)(\{M\}\otimes\{0\}+
\{0\}\otimes\{M\}+\sum\limits_{K\subset[m],K\not=[m]}
\{M_{K}\}\otimes\{M\diagup M_{K}\}) \\
=\bigtriangleup\{M\}\otimes\{0\}+
\{0\}\otimes\{0\}\otimes\{M\}+
\sum\limits_{K\subset[m],K\not=[m]}
\bigtriangleup\{M_{K}\}\otimes\{M\diagup M_{K}\},
\end{array}
$$
where

$$
\bigtriangleup\{M_{K}\} 
=\{M_{K}\}\otimes\{0\}+
\{0\}\otimes\{M_{K}\}+ 
\sum\limits_{I\subset K,\,I\not= K}
\{M_{I}\}\otimes
(\{M_{K}\diagup M_{I}\})
$$
Comparing both sides of the formula in
theorem 2.1 we reach the conclusion of
theorem 2.1. 

\end{proof}

The unit $u$ and counit $\eta$ of $\mathcal{H}_{adj}$ are defined
as follows:

\begin{equation}
u:c\mapsto c\{0\},\,\,c\in\mathbb{C},
\end{equation}

\begin{equation}
\eta:\{0\}\mapsto 1,\eta:M\mapsto 0,
for\,\,M\not=0.
\end{equation}
It is easy to check that tuple
$(\mathcal{H}_{adj},\oplus,u,\bigtriangleup,\eta)$
is a bialgebra.

Let $\overline{\mathcal{H}}_{adj}=\mathbf{ker}(\eta)$,
and $\overline{\bigtriangleup}$ denote
the reduced coproduct on $\overline{\mathcal{H}}_{adj}$,

$$
\overline{\bigtriangleup}\{M\}=
\bigtriangleup\{M\}-\{M\}\otimes\{0\}
-\{0\}\otimes\{M\}.
$$
Then we have the following conclusion.

\begin{proposition}
$\overline{\bigtriangleup}$ is conilpotent,
i.e. for any connected 
$\{M\}\in M_{adj}(\infty,\mathbb{N})$ ($M\not= 0$),
there is an positive integer $n$ such that

$$
\overline{\bigtriangleup}^{n}\{M\}=0,
$$
where $\overline{\bigtriangleup}^{n}$
is defined to be

$$
\begin{array}{ccc}
\overline{\bigtriangleup}^{n+1}=
(\overline{\bigtriangleup}\otimes &
\underbrace{1\otimes\cdots\otimes 1} &)
\overline{\bigtriangleup}^{n}. \\
 & n-times &
\end{array}
$$
\end{proposition}

According to the previous discussion we know that
$\mathcal{H}_{adj}$ is a conilpotent bialgebra,
thus, a Hopf algebra. Actually, when
the reduced coproduct is conilpotent,
the antipode $S$ can be expressed by the redeced
coproduct (see ?). Setting

$$
\begin{array}{c}
\oplus^{n}:\mathcal{H}_{adj}^{\otimes^{n}}
\rightarrow\mathcal{H}_{adj},\,\,
\oplus^{n}:\{M_{1}\}\otimes\cdots\otimes
\{M_{n}\}\mapsto\{M_{1}\}\oplus\cdots\oplus
\{M_{n}\}, \\
\{M_{i}\}\in M_{adj}(m_{i},\mathbb{N})
\diagup\sim,\,
i=1,\cdots,n,
\end{array}
$$
then the antipode $S$ will be of the following form,

$$
S(\{M\})=-\{M\}+\sum\limits_{n\geq 1}(-1)^{n}
\oplus^{n}(\overline{\bigtriangleup}^{n-1}\{M\}),
\,\,\{M\}\in M_{adj}(+\infty,\mathbb{N}).
$$

All previous discussions in this subsection 
can be generalized to the situation of the 
extended adjacency matrices. Let

$$
\mathcal{H}_{adj(e)}=\mathbf{Span}_{\mathbb{C}}
(M_{adj}(+\infty,\mathbb{N})_{(e)}).
$$
Replacing the adjacency matrices with the 
extended adjacency matrices in definition 2.4 
and in all previous conclusions, we can prove that
$\mathcal{H}_{adj(e)}$ is a Hopf algebra.

\section{Insertion of the adjacency matrices}

In this section we will discuss the insertion
of the adjacency matrices. To define the insertion
of the adjacency matrices we need to introduce
the decomposing map for the non-negative integers or
mutiple indices. The decomposing map for the 
non-negative integers is map $\iota:\mathbb{N}
\rightarrow\mathbb{N}^{l}$, 

$$
\iota:a\mapsto(a_{1},\cdots,a_{l}),\,
a,a_{1},\cdots,a_{l}\in\mathbb{N},\,
a_{1}+\cdots+a_{l}=a,
$$
where $l$ is a positive integer.
In the situation of the multiple indices,
the decomposing map can be defined in 
terms of the matrices as follows.
Let $(m_{1},\cdots,m_{k})\in\mathbb{N}^{k}$,
then

$$
\iota:
\begin{pmatrix}
m_{1} \\
\vdots \\
m_{k}
\end{pmatrix}
\mapsto
\begin{pmatrix}
a_{11} &\cdots & a_{1l} \\
\vdots &\ddots & \vdots \\
a_{k1} &\cdots & a_{kl}
\end{pmatrix},
$$
where $\iota(m_{i})=(a_{i1},\cdots,a_{il})
\in\mathbb{N}^{l}$, $a_{i1}+
\cdots+a_{il}=m_{i}$ ($i=1,\cdots,k$).

\begin{definition}
Let $M\in M_{adj}(m,\mathbb{N})$ 
(or $(M,b)\in M_{adj}(m+1,\mathbb{N})_{(e)}$),
$N\in M_{adj}(n,\mathbb{N})$,
we define the insertion
of the adjacency matrices as follows.
	
\begin{itemize}
\item $\mathbf{The\,\,situation\,\,of\,\,
N\,\,being\,\,connected:}$
Let $1\leq i\leq m$, $\iota_{i}$ be a decomposing
map,

$$
\iota_{i}:
\begin{pmatrix}
m_{1i} \\
\vdots \\
m_{i-1i}
\end{pmatrix}
\mapsto
A_{12}=
\begin{pmatrix}
a_{11} & \cdots & a_{1n} \\
\vdots & \ddots & \vdots \\
a_{i-1\,1} & \cdots & a_{i-1\,n}
\end{pmatrix},
$$

$$
\iota_{i}:
\begin{pmatrix}
m_{i+1i} \\
\vdots \\
m_{mi}
\end{pmatrix}
\mapsto
A_{32}=
\begin{pmatrix}
a_{i+1\,1} & \cdots & a_{i+1\,n} \\
\vdots & \ddots & \vdots \\
a_{m1} & \cdots & a_{mn}
\end{pmatrix},
$$
where $(m_{1i},\cdots,m_{i-1i},0,
m_{i+1i},\cdots,m_{mi})^{T}$ is the $i$th
column of $M$.
The insertion of $N$ into $M$ at $i$
by $\iota_{i}$ is an adjacency matrix in 
$M_{adj}(m+n-1,\mathbb{N})$ with the following form:
		
\begin{equation}
\begin{pmatrix}
M_{11} & A_{12} & M_{13} \\
A_{12}^{T} & N & A_{32}^{T} \\
M_{13}^{T} & A_{32} & M_{33}
\end{pmatrix},
\end{equation}
where $M_{11}=M_{I_{1}}$ ($I_{1}=\{1,\cdots,i-1\}$),
$M_{33}=M_{I_{2}}$ ($I_{2}=\{i+1,\cdots,m\}$),
		
$$
M_{13}=
\begin{pmatrix}
m_{1i+1} & \cdots & m_{1m} \\
\vdots & \ddots & \vdots \\
m_{i-1\,i+1} & \cdots &
m_{i-1\,m}
\end{pmatrix},
$$
The block matrix
(3.1) is denoted by $(N\hookrightarrow_{(i,\iota_{i})}M)$
and $i$ is called the position of the above insertion.

The stuation of the extended adjacency matrices
is similar. Let $(M,a)\in M_{adj}(m+1,\mathbb{N})_{(e)}$, 
$(N,b)\in M_{adj}(n+1,\mathbb{N})_{(e)}$,
the insertion $(N,b)$ into $(M,a)$ at $i$ by
$\iota_{i}$ is defined to be

\begin{equation}
((N,b)\hookrightarrow_{(i,\iota_{i})}(M,a))=
((N\hookrightarrow_{(i,\iota_{i})}M),
a_{I_{1}},\iota_{i}(a_{i}),a_{I_{2}}),
\end{equation}
where $\iota_{i}(a_{i})=(a_{i1},
\cdots,a_{in})$ ($a_{i1}+
\cdots+a_{in}=a_{i}$). 

We define

$$
\begin{array}{c}
(0\hookrightarrow_{i}M)=M,\,
(M\hookrightarrow 0)=M, \\
(0\hookrightarrow_{i}(M,a))=(M,a),\,
((M,a)\hookrightarrow 0)=(M,a).
\end{array}
$$

\item $\mathbf{The\,\,situation\,\,of\,\,N
\,\,being\,\,disconnected:}$
Let $N\sim N_{1}\oplus\cdots\oplus N_{k}$ ($2\leq k\leq n-1$) , 
each $N_{j}$ be connected ($j=1,\cdots,k$), then we define
the insertion of $N$ into $M$, or $(M,b)$, by 
$\iota_{i_{1}},\cdots,\iota_{i_{k}}$ at
$i_{1},\cdots,i_{k}$ as
		
\begin{equation}
(N_{k}\hookrightarrow_{(i_{k},\iota_{i_{k}})}(\cdots
(N_{2}\hookrightarrow_{(i_{2},\iota_{i_{2}})}
(N_{1}\hookrightarrow_{(i_{1},\iota_{i_{1}})}M))\cdots)),
\end{equation}
or

\begin{equation}
(N_{k}\hookrightarrow_{(i_{k},\iota_{i_{k}})}(\cdots
(N_{2}\hookrightarrow_{(i_{2},\iota_{i_{2}})}
(N_{1}\hookrightarrow_{(i_{1},\iota_{i_{1}})}(M,b)))\cdots)),
\end{equation}
where $i_{a}\not=i_{b}$ ($a\not=b$).
We denote the matrix in (3.3), or (3.4), by
		
$$
(N_{1}\oplus\cdots\oplus N_{k}\hookrightarrow
_{(i_{1},\cdots,i_{k},\iota_{i_{1}},\cdots,\iota_{i_{k}})}M),
$$
or

$$
(N_{1}\oplus\cdots\oplus N_{k}\hookrightarrow
_{(i_{1},\cdots,i_{k},\iota_{i_{1}},\cdots,\iota_{i_{k}})}(M,b)).
$$
\end{itemize}
\end{definition}

The elmentray subjects concerning the adjacency 
matrices, or the extended adjacency matrices in 
this article are connectedness, quotient and
insertion. We have seen and we will see the 
properties of the adjacency 
matrices and the extended adjacency matrices
are almost same. Thus, we will foucs on the situation
of the adjacency matrices below.

\begin{remark}
\begin{itemize}
\item For convenience, we introduce some
compact symbols about the direct sum and
insertion. Some time the direct sum
$N_{1}\oplus\cdots\oplus N_{n}$ will be denoted
by $(N)_{[n]}$ below. Similarly, for
a subset $\varLambda\subset[n]$,
the direct sum $\bigoplus_{j\in\varLambda}N_{j}$
will be denoted by $(N)_{\varLambda}$ for short.
Furthermore, the insertion

$$
(N_{k}\hookrightarrow_{(i_{k},\iota_{i_{k}})}(\cdots
(N_{2}\hookrightarrow_{(i_{2},\iota_{i_{2}})}
(N_{1}\hookrightarrow_{(i_{1},\iota_{i_{1}})}M))\cdots))
$$
will be denoted by 

$$
((N)_{[n]}\hookrightarrow_
{(i_{[n]},\iota_{i_{[n]}})}M),
$$
where $i_{[n]}=(i_{1},\cdots,i_{n})$,
$\iota_{i_{[n]}}=(\iota_{i_{1}},\cdots,\iota_{i_{n}})$.

\item If $M\in M_{adj}(m,\mathbb{N})$ is disconnected, 
then $M\sim M_{1}\oplus\cdots\oplus M_{k}$
where each $M_{i}$ is connected ($i=1,\cdots,k$).
Then, for $1\leq i\leq m$, there is $j$
($1\leq j\leq k$), such that $i$ is an index
of the rows (or columns) of $M_{j}$. By definition 3.1,
It is easy to see that

$$
(N\hookrightarrow_{(i,\iota_{i})}M)=
(M)_{[k]\setminus\{j\}}\oplus
(N\hookrightarrow_{(i,\iota_{i})}M_{j}).
$$
where $N$ is a connected adjacency matrix.
\end{itemize}	
\end{remark}

\begin{proposition}
	Let $M\in M_{adj}(m,\mathbb{N})$,
	$N\in M_{adj}(n,\mathbb{N})$, $M$ and $N$ be connected,
	$(N\hookrightarrow_{i,\iota_{i}}M)$ be the
	insertion of $N$ into $M$ at $i$ by $\iota_{i}$,
	where $1\leq i\leq m$, $\iota_{i}$ is the 
	decomposing map, $\sigma\in\mathbf{S}_{n}$,
	$\tau\in\mathbb{S}_{m}$. Then, we have
	
	\begin{equation}
	(N\hookrightarrow_{(i,\iota_{i})}M)\sim
	(\sigma(N)\hookrightarrow_{(\tau^{-1}(i),
	\sigma(\iota_{i}))}\tau(M)),
	\end{equation}
	where, based on the block matrix (3.1),
	$\sigma(\iota_{i})(m_{ji})=(a_{j\sigma(1)},
	\cdots,a_{j\sigma(n)})$.
\end{proposition}

\begin{proof}
	By definition, we know that
	$\tau(M)=(m_{\tau(a)\tau(b)})_{m\times m}$,
	the eatries on $i$th column of $M$ will
	be on $\tau^{-1}(i)$th column of $\tau(M)$.
	Pricisely, let $a=\tau^{-1}(i)$, we have
	
	$$
	\tau(M)=
	\begin{pmatrix}
	V_{11} &
	\begin{array}{c}
	m_{\tau(1)i} \\
	\vdots \\
	m_{\tau(a-1)\,i} 
	\end{array} 
	& V_{12}  \\
	m_{i\tau(1)}\cdots m_{i\tau(a-1)} &
	0 & m_{i\tau(a+1)}\cdots m_{i\tau(m)} \\
	V_{21} &
	\begin{array}{c}
	m_{\tau(a+1)\,i} \\
	\vdots \\
	m_{\tau(m)\,i} 
	\end{array}
	& V_{22}
	\end{pmatrix},
	$$
	where
	
	$$
	V_{11}=
	\begin{pmatrix}
	m_{\tau(1)\tau(1)} & \cdots & m_{\tau(1)\tau(a-1)} \\
	\vdots &\ddots & \vdots \\
	m_{\tau(a-1)\tau(1)} &\cdots & m_{\tau(a-1)\tau(a-1)}
	\end{pmatrix},
	$$
	
	$$
	V_{12}=
	\begin{pmatrix}
	m_{\tau(1)\tau(a+1)} &\cdots & m_{\tau(1)\tau(m)} \\
	\vdots & \ddots & \vdots \\
	m_{\tau(a-1)\tau(a+1)} &\cdots & m_{\tau(a-1)\tau(m)}
	\end{pmatrix},
	$$
	
	$$
	V_{22}=
	\begin{pmatrix}
	m_{\tau(a+1)\tau(a+1)} &\cdots & m_{\tau(a+1)\tau(m)} \\
	\vdots & \ddots & \vdots \\
	m_{\tau(m)\tau(a+1)} &\cdots & m_{\tau(m)\tau(m)}
	\end{pmatrix},
	$$
	$V_{21}=V_{12}^{T}$.
	If $\sigma(N)=PNP^{T}$, where $P$ is a $n\times n$
	permutation matrix, then we have
	
	$$
	(\sigma(N)\hookrightarrow_{(\tau^{-1}(i),
		\sigma(\iota_{i}))}\tau(M))=
	\begin{pmatrix}
	V_{11} & B_{12} & V_{12} \\
	B_{12}^{T} & PNP^{T} & B_{32}^{T} \\
	V_{21} & B_{32} & V_{22}
	\end{pmatrix}.
	$$
	If we express the decomposing map $\iota_{i}$
	as a matrix, i.e.
	
	$$
	A=
	\begin{pmatrix}
	a_{11} &\cdots & a_{1n} \\
	\vdots & \ddots & \vdots \\
	a_{i-1\,1} &\cdots & a_{i-1\,n} \\
	a_{i+1\,1} &\cdots & a_{i+1\,n} \\
	\vdots & \ddots & \vdots \\
	a_{m1} &\cdots & a_{mn} 
	\end{pmatrix}
	=
	\begin{pmatrix}
	A_{12} \\
	A_{32}
	\end{pmatrix},
	$$
	and let $\tau(M)=P_{1}MP_{1}^{T}$,
	where $A_{12}$ and $A_{32}$ are given in definition 3.1,
	$P_{1}$ is a $m\times m$ permutation matrix,
	then
	
	$$
	B=
	\begin{pmatrix}
	B_{12} \\
	B_{32}
	\end{pmatrix}
	=P_{1}AP^{T}=
	\begin{pmatrix}
	a_{\tau(1)\sigma(1)} &\cdots & a_{\tau(1)\sigma(n)} \\
	\vdots & \ddots & \vdots \\
	a_{\tau(a-1)\,\sigma(1)} &\cdots & a_{\tau(a-1)\,\sigma(n)} \\
	a_{\tau(a+1)\,\sigma(1)} &\cdots & a_{\tau(a+1)\,\sigma(n)} \\
	\vdots & \ddots & \vdots \\
	a_{\tau(m)\sigma(1)} &\cdots & a_{\tau(m)\sigma(n)} 
	\end{pmatrix}.
	$$
	Comparing the expression of $(\sigma(N)\hookrightarrow_{(\tau^{-1}(i),
	\sigma(\iota_{i}))}\tau(M))$ with the block matrix (3.1)
	we know that the formula (3.5) is valid.
	
\end{proof}

We denote the set $\{(\sigma(N)\hookrightarrow_{(\tau^{-1}(i),
\sigma(\iota_{i}))}\tau(M))\}_{\sigma\in\mathbf{S}_{n},
\tau\in\mathbf{S}_{m}}$ by $\{N\}\hookrightarrow_{(i,\iota_{i})}\{M\}$.
The formula (3.5) means that 

$$
\{(\sigma(N)\hookrightarrow_{(\tau^{-1}(i),
\sigma(\iota_{i}))}\tau(M))\}_{\sigma\in\mathbf{S}_{n},
\tau\in\mathbf{S}_{m}}\subset
\{(N\hookrightarrow_{i,\iota_{i}}M)\}.
$$
Thus, we do not distinguish 
$\{(\sigma(N)\hookrightarrow_{(\tau^{-1}(i),
\sigma(\iota_{i}))}\tau(M))\}_{\sigma\in\mathbf{S}_{n},
\tau\in\mathbf{S}_{m}}$ and
$\{(N\hookrightarrow_{(i,\iota_{i})}M)\}$.
In the sense of previous discussions, the insertion 
is a well defined map

$$
(\{N\},\{M\})\longrightarrow
\{(N\hookrightarrow_{(i,\iota_{i})}M)\}.
$$

\begin{proposition}
	Let $M\in M_{adj}(m,\mathbb{N})$,
	$N\in M_{adj}(n,\mathbb{N})$, both $M$ and $N$ be connected,
	then $(N\hookrightarrow_{(i,\iota_{i})}M)$ is also
	connected.
\end{proposition}

\begin{proof}
	We assume $(N\hookrightarrow_{(i,\iota_{i})}M)$ is
	disconnectd. We want to show that this assumption will result
	in a contradiction. Let $O=(N\hookrightarrow_{(i,\iota_{i})}M)$	
	for short, then $O$	adapts a decomposition
	
	$$
	O\sim O_{I_{1}}\oplus\cdots\oplus O_{I_{k}}.
	$$
	where $\{I_{i}\}_{i=1}^{k}\in\mathbf{Part}
	([m+n-1])$, each $O_{I_{i}}$ is
	connected ($i=1,\cdots,k,\,k\geq 2$).
	
	Without loss of generality, we assume $i=1$.
	Then, $(N\hookrightarrow_{(i,\iota_{i})}M)$ will be
	of the following form,
	
	$$
	(N\hookrightarrow_{(1,\iota_{1})}M)=
	\begin{pmatrix}
	N & A^{T} \\
	A & M_{I}
	\end{pmatrix},
	$$
	where $I=\{2,\cdots,m\}$. Because both $M$ and
	$N$ are connected, we know that $A\not=0$ and
	there is some $I_{i^{\prime}}$ such that 
	$[n]\subset I_{i^{\prime}}$,
	$[n]\not=I_{i^{\prime}}$.
	For simplicity, we assume $i^{\prime}=1$
	and identify $I$ with $\{n+1,\cdots,m+n-1\}$. 
	Then, we know that
	$I_{2}\cup\cdots\cup I_{k}\subset I$.
	Let
	
	$$
	\pi(O)=\mathbf{diag}(O_{I_{1}},\cdots,O_{I_{k}}),\,
	\pi\in\mathbf{S}_{m+n-1}.
	$$
	Noting $[n]\subset I_{1}$, we can
	assume
	
	$$
	O_{I_{1}}=
	\begin{pmatrix}
	N & B^{T} \\
	B & O_{I_{1}^{\prime}}
	\end{pmatrix},
	$$
	where $I_{1}^{\prime}=I_{1}\setminus[n]$.
	Above expression means that the permutation $\pi$
	keeps the positions of $1,\cdots,n$, 
	equivalently, $\pi$ is of the following form
	
	$$
	\pi=
	\begin{pmatrix}
	1&\cdots&n&n+1&\cdots&n+m-1 \\
	1&\cdots&n&\pi(n+1)&\cdots&\pi(n+m-1)
	\end{pmatrix}.
	$$
	Thus $\pi$ induces a permutation on $[m]$
	denoted by $\pi^{\prime}$,
	
	$$
	\pi^{\prime}=
	\begin{pmatrix}
	1&2&\cdots&m \\
	1&\pi(n+1)-n+1&\cdots&\pi(n+m-1)-n+1
	\end{pmatrix}.
	$$
	If $\pi^{\prime}$ corresponds to a $(m-1)\times(m-1)$
	permutation matrix $P$, then $\pi$ will
	correspond to the permutation matrix
	
	$$
	P_{1}=
	\begin{pmatrix}
	E_{n} & 0 \\
	0 & P
	\end{pmatrix},
	$$
	where $E_{n}$ denotes the unit matrix
	of order $n$. Therefore we have
	
	$$
	P_{1}OP_{1}^{T}=
	\begin{pmatrix}
	N & A^{T}P^{T} \\
	PA & PM_{I}P^{T}
	\end{pmatrix}
	=\mathbf{diag}(O_{I_{1}},\cdots,O_{I_{k}}).
	$$
	Above expression implies
	
	$$
	PM_{I}P^{T}=\mathbf{diag}(O_{I_{1}^{\prime}},
	O_{I_{2}},\cdots,O_{I_{k}}).
	$$
	By recovering $M$ from $P_{1}OP_{1}^{T}$ we know
	that
	
	$$
	M\sim M_{1}\oplus O_{I_{2}}\oplus\cdots\oplus O_{I_{k}}.
	$$
	Finally, we reach a contradiction.
	
\end{proof}

We now turn to the situation of
$\{((\bigoplus_{j}N_{j})\hookrightarrow
_{\{\ast\}}(\bigoplus_{i}M_{i}))\}$.

\begin{proposition}
Let $\{M_{i}\}$, $\{N_{j}\}$ be connected adjacency matrices
($i=1,\cdots,m,\,j=1,\cdots,n$). Then, there
a subset $I\subset[m]$ assigned to
a sequence of subsets of $[n]$, 
$\{J_{i}\}_{i\in I}\in\mathbf{Part}([n])$,
such that

\begin{equation}
\{((N)_{[n]}\hookrightarrow_
{(i_{[n]},\iota_{i_{[n]}})}
(M)_{[m]})\} 
=\{((N)_{[n]}\hookrightarrow_
{(i_{[n]},\iota_{i_{[n]}})}(M)_{I})\}
\oplus\{(M)_{I^{c}}\},
\end{equation}
and

$$
\{((N)_{[n]}\hookrightarrow_
{(i_{[n]},\iota_{i_{[n]}})}(M)_{I})\} 
=\bigoplus\limits_{i\in I}\{((N)_{J_{i}}
\hookrightarrow_{(i_{J_{i}},\iota_{i_{J_{i}}})}M_{i})\},
$$
where $I^{c}=[m]\setminus I$.

\end{proposition}

\begin{proof}
Let $M=\bigoplus_{i=1}^{m}M_{i}\in M_{adj}(p,\mathbb{N})$, 
then there is a sequence of the subsets 
$\{I_{i}\}_{1\leq i\leq m}\in\mathbf{Part}([p])$
such that $M_{I_{i}}=M_{i}$ ($1=1,\cdots,m$).	
We take

$$
I=\{i\in[m]|\,\exists j\in[n],
i_{j}\in I_{i}\}.
$$
Then, from definition 3.1 we have that

$$
\begin{array}{c}
\{((N)_{[n]}\hookrightarrow_
{(i_{[n]},\iota_{i_{[n]}})}
(M)_{[m]})\} \\
=\{((N)_{[n]}\hookrightarrow_
{(i_{[n]},\iota_{i_{[n]}})}
((M)_{I}\oplus(M)_{I^{c}}))\} \\
=\{((N)_{[n]}\hookrightarrow_
{(i_{[n]},\iota_{i_{[n]}})}
(M)_{I})\}\oplus\{(M)_{I^{c}}\}.
\end{array}
$$
If we take

$$
J_{i}=\{j\in[n]|\,i_{j}\in I_{i}\},
$$
then we we have

$$
\{((N)_{[n]}\hookrightarrow_
{(i_{[n]},\iota_{i_{[n]}})}
(M)_{I})\} \\
=\bigoplus\limits_{i\in I}\{((N)_{J_{i}}
\hookrightarrow_{(i_{J_{i}},\iota_{i_{J_{i}}})}M_{i})\}.
$$
Up to now we complete the proof of proposition.

\end{proof}

Regarding the insertion as inverse operation of the
quotient, we have the following conclusion.

\begin{proposition}
	Let $M=(m_{ij})_{m\times m}$, $N=(n_{ij})_{n\times n}$
	and $Q=(q_{ij})_{q\times q}$ be three connected
	adjacency matrices, $q=m+n-1$. Then,
	
	$$
	M\sim Q\diagup N
	$$
	if and only if there is a decompsing map
	$\iota_{i}:\{m_{ji}\}_{1\leq j\leq m,\,j\not=i}
	\to\mathbb{N}^{n}$ for some $i$ ($1\leq i\leq m$)
	such that
	
	$$
	Q\sim(N\hookrightarrow_{(i,\iota_{i})}M).
	$$
\end{proposition}

\begin{proof}
	By definition 2.3 and definition 3.1, it is 
	obvious that we have
	
	$$
	(N\hookrightarrow_{(i,\iota_{i})}M)
	\diagup N\sim M.
	$$
	
	Now we assume $M\sim Q\diagup N$, then, there is a
	subset $I\subset[q]$ ($|I|=n$) such that $Q_{I}=N$.
	Recalling definition 2.3, we have
	
	$$
	M=
	\begin{pmatrix}
	0 & q_{12}^{\ast}\cdots q_{1m}^{\ast} \\
	\begin{array}{c}
	q_{21}^{\ast} \\
	\vdots \\
	q_{m1}^{\ast}
	\end{array}
	& Q_{I^{c}}
	\end{pmatrix}.
	$$
	Without loss of generality, we assume
	$I=[n]$, then $I^{c}=\{n+1,\cdots,q\}$.
	By definition 2.3, we know that
	$q_{j1}^{\ast}=\sum_{k=1}^{n}q_{j+n-1\,k}$,
	$j=2,\cdots,m$. We can now construct the decomposing
	map in the following way.
	
	$$
	\iota_{1^{\ast}}:
	\begin{pmatrix}
	q_{21}^{\ast} \\
	\vdots \\
	q_{m1}^{\ast}
	\end{pmatrix}
	\mapsto
	\begin{pmatrix}
	q_{n+1\,1} &\cdots & q_{n+1\,n} \\
	\vdots &\ddots & \vdots \\
	q_{q1} &\cdots & q_{qn}
	\end{pmatrix}
	=Q_{21},
	$$
	then
	we have
	
	$$
	(N\hookrightarrow_{(1^{\ast},\iota_{1^{\ast}})}M)=
	\begin{pmatrix}
	N & Q_{21}^{T} \\
	Q_{21} & Q_{I^{c}}
	\end{pmatrix}
	=Q.
	$$
	
\end{proof}

In the rest of this section we will discuss
the situation to make insertion repeatly.
Let $M=(m_{ij})_{m\times m}$, $N=(n_{ij})_{n\times n}$
and $Q=(q_{ij})_{q\times q}$ be three connected
adjacency matrices. there are two possible order
to make insertion twice, which are

$$
((N\hookrightarrow_{(i,\iota_{i})}M)
\hookrightarrow_{(j,\tau_{j})}Q)\,\, and
\,\,(N\hookrightarrow_{(a,\mu_{a})}
(M\hookrightarrow_{(b,\nu_{b})}Q)).
$$
Actually, we are interested in the situation of
$(N\hookrightarrow_{(a,\mu_{a})}
(M\hookrightarrow_{(b,\nu_{b})}Q))$ which is more
complicated than other.
In this situation, there is a subset 
$I\subset[m+q-1]$ such that
$(M\hookrightarrow_{(b,\nu_{b})}Q)_{I}=M$.
For the index $a$, there two possibilities
which are $a\in I$ or $a\notin I$. When
$a\in I$, it is easy to know that $a$ corresponds
to an index $a^{\prime}$ of the row (or column)
of $M$. Thus, in this situation, we say  
$a\in M$. Similarly, when $a\notin I$,
we say $a\notin M$. We have the 
following conclusion.

\begin{lemma}
Let $M=(m_{ij})_{m\times m}$, $N=(n_{ij})_{n\times n}$
and $Q=(q_{ij})_{q\times q}$ be three connected
adjacency matrices. About the insertion $(N\hookrightarrow
_{(a,\mu_{a})}(M\hookrightarrow_{(b,\nu_{b})}Q))$,
we have

$$
(N\hookrightarrow
_{(a,\mu_{a})}(M\hookrightarrow_{(b,\nu_{b})}Q)) 
=\left\{
\begin{array}{cc}
((N\hookrightarrow_{(a^{\prime},\mu_{a^{\prime}}^{\prime})}M)
\hookrightarrow_{(b,\nu_{b}^{\prime})}Q), & a\in M, \\
(N\oplus N\hookrightarrow_{(a,\mu_{a}),(b,\nu_{b})}Q),
&  a\notin M.
\end{array}\right.
$$
\end{lemma}

\begin{proof}
Here we focus on the situation of $a\in M$.	
If $a\in M$, it is easy to see that
$a$ correspondes to an index $a^{\prime}$ of the rows 
(or columns) of $M$. For simplicity, we assume $a=b=1$. 
Then we have

$$
(M\hookrightarrow_{(1,\nu_{1})}Q)=
\begin{pmatrix}
M & A^{T} \\
A & Q_{1}
\end{pmatrix},
$$ 
and

$$
O=(N\hookrightarrow_{(1,\mu_{1})}
(M\hookrightarrow_{(1,\nu_{1})}Q))=
\begin{pmatrix}
N & B_{21}^{T} & B_{31}^{T} \\
B_{21} & M_{1} & B_{32}^{T} \\
B_{31} & B_{32} & Q_{1}
\end{pmatrix},
$$
where we let $O$ denote
$(N\hookrightarrow_{(1,\mu_{1})}
(M\hookrightarrow_{(1,\nu_{1})}Q))$ for short,
$M_{1}=(m_{ij})_{2\leq i,j\leq m}$,
$Q_{1}=(q_{ij})_{2\leq i,j\leq q}$.
The decomposing map $\nu_{1}$ is given by
$\nu_{1}:(q_{21},\cdots,q_{2q})^{T}\mapsto A
=(a_{ij})_{(q-1)\times m}$.
The decomposing map $\mu_{1}$ 
is given by

$$
\mu_{1}:
\begin{pmatrix}
(m_{21},\cdots,m_{m1})^{T} \\
(a_{21},\cdots,a_{q1})^{T}
\end{pmatrix}
\mapsto
\begin{pmatrix}
B_{21} \\
B_{31}
\end{pmatrix}
$$
If we take $I=[m+n-1]$ regarded as a
subset of $[m+n+q-2]$, then $O_{I}$
should be of the following form.

$$
O_{I}=
\begin{pmatrix}
N & B_{21}^{T} \\
B_{21} & M_{1}
\end{pmatrix}.
$$ 
Thus, $O_{I}=(N\hookrightarrow_{(1,\mu_{1}^{\prime})}M)$,
where $\mu_{1}^{\prime}=\mu_{1}
|_{\{m_{i1}\}_{2\leq i\leq m}}$, precisely, we have

$$
\mu_{1}^{\prime}:
\begin{pmatrix}
m_{21} \\
\vdots \\
m_{m1}
\end{pmatrix}
\mapsto B_{21}.
$$
The decomposing map $\nu_{1}^{\prime}$ should be
of the following form

$$
\nu_{1}^{\prime}:
\begin{pmatrix}
q_{21} \\
\vdots \\
q_{q1}
\end{pmatrix}
\mapsto 
\begin{pmatrix}
B_{31}, B_{32}
\end{pmatrix}.
$$

\end{proof}

Similarly, we can prove a more general conclusion as follows.

\begin{proposition}
Let $M_{i}$, $N_{j}$ and $Q$ be connected adjacency matrices
($i=1,\cdots,m,\,j=1,\cdots,n,\,n\geq m$).
Then we have

\begin{equation}
\begin{array}{c}
((N)_{[n]}\hookrightarrow_
{(a_{[n]},\iota_{a_{[n]}})}((M)_{[m]}
\hookrightarrow_{(q_{[m]},\tau_{q_{[m]}})}Q)) \\
=((N)_{\varLambda^{c}}\oplus((N)_{\varLambda}\hookrightarrow_
{(i_{\varLambda}^{\prime},\iota_{i_{\varLambda}^{\prime}}^{\prime})}
(M)_{\varGamma})\oplus(M)_{\varGamma^{c}})\\
\hookrightarrow_
{(q_{\varLambda^{c}},\iota_{q_{\varLambda^{c}}})\cup
(q_{\varGamma},\tau_{q_{\varGamma}}^{\prime})\cup
(q_{\varGamma^{c}},\tau_{q_{\varGamma^{c}}})}Q),
\end{array}
\end{equation}
where,

$$
\varLambda=\{j\in[n]|\,\exists i\in[m],\,
s.t.\,\,a_{j}\in M_{i}\},\,\,
\varLambda^{c}=[n]\setminus\varLambda,
$$

$$
\varGamma=\{i\in[m]|\,\exists j\in\varLambda,
a_{j}\in M_{i}\},\,\,
\varGamma^{c}=[m]\setminus\varGamma.
$$
\end{proposition}

\begin{proof}
Observong the insertion

$$
((N)_{[n]}\hookrightarrow_
{(a_{[n]},\iota_{a_{[n]}})}((M)_{[m]}
\hookrightarrow_{(q_{[m]},\tau_{q_{[m]}})}Q)),
$$
by definition 3.1, we know that for each
$N_{j}$ ($1\leq j\leq n$), there are two
possibilities which are $a_{j}\in M_{i}$
for some $i$, or, $a_{j}\notin M_{i}$ for any 
$i$ ($1\leq i\leq m$). Thus, we have a decomposition
$[n]=\varGamma\cup\varGamma^{c}$, where

$$
\varLambda=\{j\in[n]|\,\exists i\in[m],\,
s.t.\,\,a_{j}\in M_{i}\}.
$$
Similarly, for each $M_{i}$, there are two
possibilities of $i$, there is some $j$ such that
$a_{j}\in M_{i}$, or $a_{j}\notin M_{i}$ for
any $j\in[n]$. We can take

$$
\varGamma=\{i\in[m]|\,\exists j\in
[n],a_{j}\in M_{i}\}.
$$

When $j\in\varLambda^{c}$, $N_{j}$ inserts into $Q$,
thus, $i_{j}$ will be assigned to some $q_{j}$,
where $q_{j}$ is an index of the rows (or columns)
of $Q$. By definition 3.1, it is easy to see that,

$$
\begin{array}{c}
((N)_{[n]}\hookrightarrow_
{(a_{[n]},\iota_{a_{[n]}})}((M)_{[m]}
\hookrightarrow_{(q_{[m]},\tau_{q_{[m]}})}Q)) \\
=((N)_{\varLambda}\hookrightarrow_
{(a_{\varLambda},\iota_{a_{\varLambda}})}
(((N)_{\varLambda^{c}}\oplus(M)_{\varGamma}
\oplus(M)_{\varGamma^{c}})\hookrightarrow_
{(q_{\varLambda^{c}},\iota_{q_{\varLambda^{c}}})
\cup(q_{[m]},\tau_{q_{[m]}})}Q)) \\
=((N)_{\varLambda}\hookrightarrow_
{(a_{\varLambda},\iota_{a_{\varLambda}})}
((M)_{\varGamma}\hookrightarrow_
{(q_{\varGamma},\tau_{q_{\varGamma}})}O)),
\end{array}
$$
where

$$
O=(((N)_{\varLambda^{c}}\oplus(M)_{\varGamma^{c}})
\hookrightarrow_{(q_{\varLambda^{c}},\iota_{q_{\varLambda^{c}}})
\cup(q_{\varGamma^{c}},\tau_{q_{\varGamma^{c}}})}Q).
$$

When $j\in\varLambda$, there some 
$i^{\prime}\in[m]$
such that $N_{j}$ inserts into $M_{i^{\prime}}$ at $i_{j}$, 
thus $i_{j}$ will be assigned to
some $i_{j}^{\prime}$, where $i_{j}^{\prime}$ is
an index of the rows (or columns) of $M_{i^{\prime}}$.
In a way which is similar to one in the proof
of lemma 3.1, we can prove that

$$
\begin{array}{c}
(N_{j}\hookrightarrow_{(a_{j},\iota_{a_{j}})}
(M_{i^{\prime}}\hookrightarrow_
{(q_{i^{\prime}},\tau_{q_{i^{\prime}}})}
((M)_{[m]\setminus\{i^{\prime}\}}\hookrightarrow_
{(q_{i},\tau_{q_{i}})_{i\in([m]\setminus\{i^{\prime}\})}}O))) \\
=((N_{j}\hookrightarrow_
{(i_{j}^{\prime},\iota^{\prime}_{i_{j}^{\prime}})}
M_{i^{\prime}})\hookrightarrow_{(q_{i^{\prime}},\tau_{q_{i^{\prime}}}^{\prime})}
((M)_{[m]\setminus\{i^{\prime}\}}\hookrightarrow_
{(q_{i},\tau_{q_{i}})_{i\in([m]\setminus\{i^{\prime}\})}}O))
\end{array}
$$
Repeating above argument, we can prove the formula (3.7).

\end{proof}

\begin{proposition}
Let $M_{i}$, $N_{j}$ and $Q_{k}$
be connected adjacency matrices 
($i=1,\cdots,m,\,j=1,\cdots,n,\,n\geq m,
\,k=1,\cdots,q$). Then we have

\begin{equation}
\begin{array}{c}
((N)_{[n]}\hookrightarrow_
{(a_{[n]},\iota_{a_{[n]}})}
((M)_{[m]}\hookrightarrow_
{(q_{[m]},\tau_{q_{[m]}})}
(Q)_{[q]})) \\
=(Q)_{\varXi_{1}}\oplus((N)_{\varLambda_{1}}\hookrightarrow_
{(q_{\varLambda_{1}},\iota_{q_{\varLambda_{1}}})}
(Q)_{\varXi_{2}})\oplus((M)_{\varGamma_{1}}
\hookrightarrow_{(q_{\varGamma_{1}},
\tau_{q_{\varGamma_{1}}})}
(Q)_{\varXi_{3}})  
\oplus O_{\varLambda,\varGamma,\varXi},
\end{array}
\end{equation}
where

$$
\begin{array}{c}
O_{\varLambda,\varGamma,\varXi} 
=((N)_{\varLambda_{2}}\oplus((N)_{\varLambda_{3}}
\hookrightarrow_{(i_{\varLambda_{3}},\kappa_{i_{\varLambda_{3}}})}
(M)_{\varGamma_{2}})\oplus(M)_{\varGamma_{3}} \\
\hookrightarrow_
{(q_{\varLambda_{2}},\lambda_{q_{\varLambda_{2}}})
\cup(q_{\varGamma_{2}},\gamma_{q_{\varGamma_{2}}})
\cup(q_{\varGamma_{3}},\gamma_
{q_{\varGamma_{3}}})}(Q)_{\varXi_{4}}),
\end{array}
$$
and $\{\varLambda_{1},\varLambda_{2},
\varLambda_{3}\}\in\mathbf{Part}([n])$,
$\{\varGamma_{1},\varGamma_{2},\varGamma_{3}\}\in
\mathbf{Part}([m])$,
$\{\varXi_{1},\varXi_{2},\varXi_{3},\varXi_{4}\}
\in\mathbf{Part}([q])$.
\end{proposition}
	
\begin{proof}
The proof of the formula (3.8) concerns the 
decomposition of $((N)_{[n]}\hookrightarrow_
{(a_{[n]},\iota_{a_{[n]}})}
((M)_{[m]}\hookrightarrow_
{(q_{[m]},\tau_{q_{[m]}})}
(Q)_{[q]}))$ according to the way
$N_{[n]}$ and $M_{[m]}$
insert into $Q_{\underline{q}}$, 
thus, concerns the decomposition
of $[n]$, $[m]$ and $[q]$. Firstly, 
recalling the formula (3.4),
we know that there is a obvious
decomposition of $[q]$,
$[q]=\varXi\cup\varXi^{c}$, where
$\varXi^{c}=[q]\setminus\varXi$, and 

$$
\varXi=\{k\in[q]|\,\exists 
q_{i}\,\,s.t.\,\,q_{i}\in Q_{k}\}.
$$
Thus we have

$$
\begin{array}{c}
((M)_{[m]}\hookrightarrow_
{(q_{[m]},\tau_{q_{[m]}})}
(Q)_{[q]}) \\
=((M)_{[m]}\hookrightarrow_
{(q_{[m]},\tau_{q_{[m]}})}
((Q)_{\varXi}\oplus(Q)_{\varXi^{c}})) \\
=((M)_{[m]}\hookrightarrow_
{(q_{[m]},\tau_{q_{[m]}})}
(Q)_{\varXi})\oplus(Q)_{\varXi^{c}}.
\end{array}
$$

Simiarly, the decomposition 
$[q]=\varXi\cup\varXi^{c}$ will
induces a decomposition of $[n]$,
$[n]=\varLambda\cup\varLambda^{c}$,
such that

$$
\begin{array}{c}
((N)_{[n]}\hookrightarrow_
{(a_{[n]},\iota_{a_{[n]}})}
((M)_{[m]}\hookrightarrow_
{(q_{[m]},\tau_{q_{[m]}})}
(Q)_{[q]})) \\
=((N)_{[n]}\hookrightarrow_
{(a_{[n]},\iota_{a_{[n]}})}
(((M)_{[m]}\hookrightarrow_
{(q_{[m]},\tau_{q_{[m]}})}
(Q)_{\varXi})\oplus(Q)_{\varXi^{c}})) \\
=((N)_{\varLambda}\hookrightarrow_
{(a_{\varLambda},\iota_{a_{\varLambda}})}
((M)_{[m]}\hookrightarrow_
{(q_{[m]},\tau_{q_{[m]}})}
(Q)_{\varXi}))
\oplus
((N)_{\varLambda^{c}}\hookrightarrow_
{(a_{\varLambda^{c}},\iota_{a_{\varLambda^{c}}})}
(Q)_{\varXi^{c}}).
\end{array}
$$
With the help of the formula (3.6) once more,
we have

$$
\begin{array}{c}
((N)_{\varLambda}\hookrightarrow_
{(a_{\varLambda},\iota_{a_{\varLambda}})}
((M)_{[m]}\hookrightarrow_
{(q_{[m]},\tau_{q_{[m]}})}
(Q)_{\varXi})) \\
=((N)_{\varLambda}\hookrightarrow_
{(a_{\varLambda},\iota_{a_{\varLambda}})}
(((M)_{\varGamma}\hookrightarrow_
{(q_{\varGamma},\tau_{q_{\varGamma}})}
(Q)_{\varXi^{\prime}})
\oplus((M)_{\varGamma^{c}}\hookrightarrow_
{(q_{\varGamma^{c}},\tau_{q_{\varGamma^{c}}})}
(Q)_{\varXi^{\prime\prime}}))) \\
=((N)_{\varLambda}\hookrightarrow_
{(a_{\varLambda},\iota_{a_{\varLambda}})}
((M)_{\varGamma}\hookrightarrow_
{(q_{\varGamma},\tau_{q_{\varGamma}})}
(Q)_{\varXi^{\prime}}))
\oplus((M)_{\varGamma^{c}}\hookrightarrow_
{(q_{\varGamma^{c}},\tau_{q_{\varGamma^{c}}})}
(Q)_{\varXi^{\prime\prime}}),
\end{array}
$$
In fact, by definition of $\varXi$, we know that
$\varXi$ induces a decomposition of
$[m]$, $\{I_{k}\}_{k\in\varXi}\in
\mathbf{Part}([m])$, where

$$
I_{k}=\{i\in[m]|\,q_{i}\in Q_{k}\},\,
k\in\varXi.
$$
Then, $\varXi^{\prime\prime}$ is able to be taken as

$$
\varXi^{\prime\prime}=\{k\in\varXi|\,a_{j}\notin
((M)_{I_{k}}\hookrightarrow_
{(q_{I_{k}},\tau_{q_{I_{k}}})}Q_{k}),
\forall j\in\varLambda\},
$$
$\varXi^{\prime}=\varXi\setminus\varXi^{\prime\prime}$.
Moreover, we have $\varGamma=
\bigcup_{k\in\varXi^{\prime}}I_{k}$,
$\varGamma^{c}=[m]\setminus\varGamma$.

Similarly, we have

$$
((N)_{\varLambda^{c}}\hookrightarrow_
{(a_{\varLambda^{c}},\iota_{a_{\varLambda^{c}}})}
(Q)_{\varXi^{c}})
=((N)_{\varLambda^{c}}\hookrightarrow_
{(a_{\varLambda^{c}},\iota_{a_{\varLambda^{c}}})}
(Q)_{\varXi_{c,N\hookrightarrow Q}})\oplus
(Q)_{\varXi_{c,Q}}.
$$

We now pay attention to the term
$((N)_{\varLambda}\hookrightarrow_
{(a_{\varLambda},\iota_{a_{\varLambda}})}
((M)_{\varGamma}\hookrightarrow_
{(q_{\varGamma},\tau_{q_{\varGamma}})}
(Q)_{\varXi^{\prime}}))$.
The decomposition 

$$
((M)_{\varGamma}\hookrightarrow_
{(q_{\varGamma},\tau_{q_{\varGamma}})}
(Q)_{\varXi^{\prime}})
=\bigoplus\limits_{k\in\varXi^{\prime}}
((M)_{I_{k}}\hookrightarrow_
{(q_{I_{k}},\tau_{q_{I_{k}}})}Q_{k})
$$
induces a decomposition of $\varLambda$,
which is $\{J_{k}\}_{k\in\varXi^{\prime}}$,
where

$$
J_{k}=\{j\in\varLambda|\,a_{j}\in
((M)_{I_{k}}\hookrightarrow_
{(q_{I_{k}},\tau_{q_{I_{k}}})}Q_{k})\}.
$$
By definiton of $\varXi^{\prime}$, it is
easy to see that $I_{k}\not=\emptyset$,
and $J_{k}\not=\emptyset$ ($k\in\varXi^{\prime}$),
and

$$
\begin{array}{c}
((N)_{\varLambda}\hookrightarrow_
{(a_{\varLambda},\iota_{a_{\varLambda}})}
((M)_{\varGamma}\hookrightarrow_
{(q_{\varGamma},\tau_{q_{\varGamma}})}
(Q)_{\varXi^{\prime}})) \\
=\bigoplus\limits_{k\in\varXi^{\prime}}
((N)_{J_{k}}\hookrightarrow_{(a_{J_{k}},\iota_{a_{J_{k}}})}
((M)_{I_{k}}\hookrightarrow_
{(q_{I_{k}},\tau_{q_{I_{k}}})}Q_{k})).
\end{array}
$$

Noting the formula (3.7), we have

$$
\begin{array}{c}
((N)_{J_{k}}\hookrightarrow_{(a_{J_{k}},\iota_{a_{J_{k}}})}
((M)_{I_{k}}\hookrightarrow_
{(q_{I_{k}},\tau_{q_{I_{k}}})}Q_{k})) \\
=((N)_{J_{k}^{\prime\prime}}\oplus((N)_{J_{k}^{\prime}}
\hookrightarrow_{(i_{J_{k}^{\prime}},
\kappa_{i_{J_{k}^{\prime}}})}(M)_{I_{k}^{\prime}})
\oplus(M)_{I_{k}^{\prime\prime}}
\hookrightarrow_{(q_{J_{k}^{\prime\prime}},
\iota_{q_{J_{k}^{\prime\prime}}})\cup
(q_{I_{k}},\tau^{\prime}_{q_{I_{k}}})}Q_{k}),
\end{array}
$$
where $k\in\varXi^{\prime}$, $J_{k}=J_{k}^{\prime}\cup
J_{k}^{\prime\prime}$, $J_{k}^{\prime}\cap
J_{k}^{\prime\prime}\not=\emptyset$,
$I_{k}=I_{k}^{\prime}\cup
I_{k}^{\prime\prime}$, $I_{k}^{\prime}\cap
I_{k}^{\prime\prime}\not=\emptyset$.
If we take $\varLambda^{\prime}=
\bigcup_{k\in\varXi^{\prime}}J_{k}^{\prime}$,
$\varLambda^{\prime\prime}=
\bigcup_{k\in\varXi^{\prime}}J_{k}^{\prime\prime}$,
$\varGamma^{\prime}=
\bigcup_{k\in\varXi^{\prime}}I_{k}^{\prime}$,
$\varGamma^{\prime\prime}=
\bigcup_{k\in\varXi^{\prime}}I_{k}^{\prime\prime}$
then we have

$$
\begin{array}{c}
((N)_{\varLambda}\hookrightarrow_
{(a_{\varLambda},\iota_{a_{\varLambda}})}
((M)_{\varGamma}\hookrightarrow_
{(q_{\varGamma},\tau_{q_{\varGamma}})}
(Q)_{\varXi^{\prime}})) \\
=(((N)_{\varLambda^{\prime\prime}}\oplus
((N)_{\varLambda^{\prime}}\hookrightarrow_
{(i_{\varLambda^{\prime}},\kappa_{i_{\varLambda^{\prime}}})}
(M)_{\varGamma^{\prime}})\oplus(M)_{\varGamma^{\prime\prime}})
\hookrightarrow_{(q_{\varLambda^{\prime\prime}},
\iota_{q_{\varLambda^{\prime\prime}}})
\cup(q_{\varGamma},\tau^{\prime}_{q_{\varGamma}})}
(Q)_{\varXi^{\prime}}).
\end{array}
$$	
Summarizing the previous discussions, we can
reach the formula (3.8).	
	
\end{proof}

\begin{remark}
In the formula (3.8), we can take
$\varLambda=\varLambda_{3}$, 
$\varLambda^{c}=\varLambda_{1}\cup\varLambda_{2}$,
$\varGamma=\varGamma_{2}$,
$\varGamma^{c}=\varGamma_{1}\cup\varGamma_{3}$,
$\varXi=\varXi_{2}\cup\varXi_{3}\cup\varXi_{4}$,
$\varXi^{c}=\varXi_{1}$, then we have

$$
\begin{array}{c}
((N)_{[n]}\hookrightarrow_
{(a_{[n]},\iota_{a_{[n]}})}
((M)_{[m]}\hookrightarrow_
{(q_{[m]},\tau_{q_{[m]}})}
(Q)_{[q]})) \\
=((N)_{\varLambda^{c}}\oplus((N)_{\varLambda}
\hookrightarrow_{(i_{\varLambda},\kappa_{i_{\varLambda}})}
(M)_{\varGamma})\oplus(M)_{\varGamma^{c}} \\
\hookrightarrow_
{(q_{\varLambda^{c}},\lambda_{q_{\varLambda^{c}}})
\cup(q_{\varGamma},\gamma_{q_{\varGamma}})
\cup(q_{\varGamma^{c}},\gamma_
{q_{\varGamma^{c}}})}(Q)_{\varXi})
\oplus(Q)_{\varXi^{c}}.
\end{array}
$$
\end{remark}

\section{The algebraic structure of $\mathcal{H}_{adj}^{\ast}$}

\subsection{Basic notations and the primitive elements}

Let 

$$
\mathcal{H}_{adj,n}=\mathbf{Span}_{\mathbb{C}}
\{\{M\}\in M_{adj}(+\infty,\mathbb{N})|
\mathbf{deg}\{M\}=n\},\,n\geq 0,
$$
where $\mathcal{H}_{adj,0}=\mathbb{C}\{0\}
\cong\mathbb{C}$.
Then each $\mathcal{H}_{adj,n}$ is finite
dimensional, and we have

$$
\mathcal{H}_{adj}=\bigoplus_{n=0}^{+\infty}
\mathcal{H}_{adj,n}.
$$
For $\{M_{i}\}\in\mathcal{H}_{adj,n_{i}}$
($i=1,2$), we have 

$$
\{M_{1}\}\oplus\{M_{2}\}\in\mathcal{H}_{adj,n_{1}+n_{2}}.
$$
On the other hand, it is easy to
check that about coproduct we have

$$
\bigtriangleup:\mathcal{H}_{adj,n}
\longrightarrow\bigoplus\limits_{p+q=n}
\mathcal{H}_{adj,p}\otimes\mathcal{H}_{adj,q}.
$$
Therefore, $\mathcal{H}_{adj}$ is a connected graded
Hopf algebra (see ?).

In this section we will discuss the dual
Hopf algebra in the following sense

\begin{equation}
\mathcal{H}_{adj}^{\ast}=\bigoplus_{n=0}^{+\infty}
\mathcal{H}_{adj,n}^{\ast}.
\end{equation}

It is well known that, by definition, the 
coproduct on $\mathcal{H}_{adj}^{\ast}$
is dual to the product on $\mathcal{H}_{adj}$,
i.e. for $f\in\mathcal{H}_{adj}^{\ast}$ we have

$$
<\bigtriangleup f,\{M_{1}\}\otimes\{M_{2}\}>
=<f,\{M_{1}\}\oplus\{M_{2}\}>,
$$
where $\{M_{1}\},\{M_{2}\}\in M_{adj}(+\infty,\mathbb{N})$.
Similarly, the product on $\mathcal{H}_{adj}^{\ast}$
is dual to the coproduct on $\mathcal{H}_{adj}$.
Thus, for $f,g\in\mathcal{H}_{adj}^{\ast}$
and $\{M\}\in M_{adj}(+\infty,\mathbb{N})$
we have

$$
<f\bullet g,\{M\}>=<f\otimes g,\bigtriangleup\{M\}>,
$$
where $\bullet$ denotes the product on 
$\mathcal{H}_{adj}^{\ast}$. Because the coproduct
on $\mathcal{H}_{adj}$ is not co-commutative, thus
the multiplication  $\bullet$ is not commutative.

Let  

$$
\{f_{\{M\}}|\{M\}\in M_{adj}(+\infty,\mathbb{N}),\{M\}\not=0\}
$$
denote the set of dual bases of
$\mathcal{H}_{adj}^{\ast}$, which means
each $f_{\{M\}}$ ($\{M\}\not=0$) satisfies

$$
<f_{\{M\}},\{N\}>=\left\{
\begin{array}{cc}
1, & \{N\}=\{M\}, \\
0, & others.
\end{array}
\right.
$$
About dual bases mentioned above we have,

\begin{proposition}
Let $\{M\}\in M_{adj}(+\infty,\mathbb{N})$,
($\{M\}\not=0$), $\{M\}=\bigoplus_{i=1}^{k}\{M_{i}\}$,
each $\{M_{i}\}$ be connected ($i=1,\cdots,k$).
Then we have

\begin{equation}
\bigtriangleup f_{\{M\}}=f_{\{M\}}\otimes\eta+\eta\otimes f_{\{M\}}
+\sum\limits_{I\subset[k],I\not=I,\emptyset}
f_{\bigoplus_{i\in I}\{M_{i}\}}\otimes 
f_{\bigoplus_{i\in I^{c}}\{M_{i}\}},
\end{equation}
where $\eta$ is the co-unit on $\mathcal{H}_{adj}$,
$I^{c}=[k]\setminus I$.
\end{proposition}

\begin{proof}
Recalling the definition of $f_{\{M\}}$,

$$
<f_{\{M\}},\{N\}>=\left\{
\begin{array}{cc}
1, & \{N\}=\{M\}, \\
0, & others,
\end{array}
\right.
$$
we know that when $\{N_{1}\}\oplus\{N_{2}\}=\{M\}$,

$$
<\bigtriangleup f_{\{M\}},\{N_{1}\}\otimes\{N_{2}\}>
=<f_{\{M\}},\{N_{1}\}\oplus\{N_{2}\}>=
<f_{\{M\}},\{M\}>\not=0,
$$
otherwise,

$$
<\bigtriangleup f_{\{M\}},\{N_{1}\}\otimes\{N_{2}\}>=0.
$$
The condition $\{N_{1}\}\oplus\{N_{2}\}=\{M\}$ means
that $\{N_{1}\}=\bigoplus_{i\in I}\{M_{i}\}$,
$\{N_{2}\}=\bigoplus_{i\in I^{c}}\{M_{i}\}$ for
some subset $I\subset[k]$. Therefore,
it is natural that $\bigtriangleup f_{\{M\}}$
should be of the form

$$
\bigtriangleup f_{\{M\}}=\sum\limits_
{I\subset[k]}g_{I}\otimes h_{I^{c}},
$$
where $g_{I},h_{I^{c}}\in\mathcal{H}_{adj}^{\ast}$
satisfying

$$
\begin{array}{c}
<g_{I}\otimes h_{I^{c}},\{N_{1}\}\otimes\{N_{2}\}> 
=<g_{I},\{N_{1}\}><h_{I^{c}},\{N_{2}\}> \\
=\left\{
\begin{array}{cc}
1, & \{N_{1}\}=\bigoplus_{i\in I}\{M_{i}\},
\{N_{2}\}=\bigoplus_{i\in I^{c}}\{M_{i}\}, \\
0 & others.
\end{array}
\right.
\end{array}
$$
Thus, $g_{I}$ and $h_{I^{c}}$ will be
$f_{\bigoplus_{i\in I}\{M_{i}\}}$ and
$f_{\bigoplus_{i\in I^{c}}\{M_{i}\}}$
respectively. Particularly, when $I=\emptyset$,
$g_{I}=\eta$, when $I^{c}=\emptyset$,
$h_{I^{c}}=\eta$.

\end{proof}

\begin{corollary}
Let $\{M\}\in M_{adj}(+\infty,\mathbb{N})$,
then $\{M\}$ is connected if and only if

$$
\bigtriangleup f_{\{M\}}=
f_{\{M\}}\otimes\eta+\eta\otimes f_{\{M\}}.
$$
\end{corollary}

Let $f\in\mathcal{H}_{adj}^{\ast}$, it is well
known that, by the definition, if $f$ satisfies

$$
\bigtriangleup f=f\otimes\eta+\eta\otimes f,
$$
then it is called a primitive element in
$\mathcal{H}_{adj}^{\ast}$. Let
$\mathbf{P}(\mathcal{H}_{adj}^{\ast})$ denote
the set of all primitive elements of 
$\mathcal{H}_{adj}^{\ast}$. Then, with the help
of corollary 4.1, we have

\begin{equation}
\mathbf{P}(\mathcal{H}_{adj}^{\ast})=
\mathbf{Span}_{\mathbb{C}}(\{f_{\{M\}}
|\{M\}\,\,is\,\,connected\}).
\end{equation}

\subsection{The product on $\mathcal{H}_{adj}^{\ast}$}

About the product on $\mathcal{H}_{adj}^{\ast}$
we have the following formula.

\begin{proposition}
Let $M\in M_{adj}(m,\mathbb{N})$, 
$N\in M_{adj}(n,\mathbb{N})$ be two connected adjacency 
matrices. Then, we have

\begin{equation}
f_{\{N\}}\bullet f_{\{M\}}=\sum\limits_{i,\iota_{i}}
f_{\{(N\hookrightarrow_{i,\iota_{i}}M)\}}+
f_{\{M\}\oplus\{N\}}.
\end{equation}
\end{proposition}

\begin{proof}
By the definition, the product 
$f_{\{N\}}\bullet f_{\{M\}}$
is defined by the following formula,

$$
<f_{\{N\}}\bullet f_{\{M\}},\{Q\}>
=<f_{\{N\}}\otimes f_{\{M\}},
\bigtriangleup\{Q\}>,\,\,\{Q\}\in
M_{adj}(+\infty,\mathbb{N}).
$$
It is easy to see that when $\{Q\}$ is
connected, the meaningful choice of
$\{Q\}$ shouls be $\{(N\hookrightarrow_{i,\iota_{i}}M)\}$.
Actually, we have

$$
\begin{array}{c}
\bigtriangleup\{(N\hookrightarrow_{i,\iota_{i}}M)\} \\
=\{(N\hookrightarrow_{i,\iota_{i}}M)\}\otimes 0+
0\otimes\{(N\hookrightarrow_{i,\iota_{i}}M)\}+
\cdots+\{N\}\otimes\{(N\hookrightarrow_{i,\iota_{i}}M)\diagup N\}
+\cdots.
\end{array}
$$
Thus

$$
<f_{\{N\}}\bullet f_{\{M\}}
,\{(N\hookrightarrow_{i,\iota_{i}}M)\}>
=<f_{\{N\}}\otimes f_{\{M\}},\{N\}\otimes\{M\}>=1.
$$
In the situation of $\{Q\}$ being disconnected,
the suitable choice of $\{Q\}$ should be
$\{N\}\oplus\{M\}$. It is obvious that

$$
<f_{\{N\}}\bullet f_{\{M\}},\{N\}\oplus\{M\}>=1.
$$
For other $\{Q\}$, we have

$$
<f_{\{N\}}\bullet f_{\{M\}},\{Q\}>=0.
$$
Up to now, we have proved the formula (4.4).

\end{proof}

Furthermore, we have a more general formula
about the product on $\mathcal{H}_{adj}^{\ast}$.

\begin{theorem}
Let $M_{i}$, $N_{j}$ be connected adjacency matrices
($i=1,\cdots,k,\,j=1,\cdots,l$). Then we have

\begin{equation}
\begin{array}{c}
f_{\{(N)_{[n]}\}}\bullet 
f_{\{(M)_{[m]}\}} \\
=\sum\limits_{\varLambda\subset[n],
\varLambda\not=\emptyset}\,\,
\sum\limits_{(i_{\varLambda},\iota_{i_{\varLambda}})}\, 
f_{\{(N)_{\varLambda^{c}}\}\oplus
\{((N)_{\varLambda}\hookrightarrow_
{(i_{\varLambda},\iota_{i_{\varLambda}})}
(M)_{[m]})\}}
+f_{\{(N)_{[n]}\}
\oplus\{(M)_{[m]}\}},
\end{array}
\end{equation}
where $\varLambda^{c}=[n]\setminus\varLambda$.
\end{theorem}

\begin{proof}
Recalling the definition of the product on
$\mathcal{H}_{adj}^{\ast}$, we have

$$
\begin{array}{c}
<f_{\{(N)_{[n]}\}}\bullet 
f_{\{(M)_{[m]}\}},\{Q\}> \\
=<f_{\{(N)_{[n]}\}}\otimes 
f_{\{(M)_{[m]}\}},\bigtriangleup\{Q\}>,
\,\,\{Q\}\in M_{adj}(+\infty,\mathbb{N}).
\end{array}
$$
In order to prove theorem 4.1, we need to choose
$\{Q\}$ such that 

$$
<f_{\{(N)_{[n]}\}}\otimes 
f_{\{(M)_{[m]}\}},
\bigtriangleup\{Q\}>\not=0.
$$
Here we are interested in the situation of $m\geq 2$.
Hence, $\{Q\}$ should be disconnected. Actually,
if $\{Q\}=\{Q_{1}\}\oplus\cdots\oplus\{Q_{p}\}$,
where each $\{Q_{k}\}$ is connected ($k=1,\cdots,p$),
then $p\geq m$. We focus on the right factors in
the tensor, then $\bigtriangleup\{Q_{k}\}$ 
($i=1,\cdots,p$) will be required to provide 
$\{M_{i}\}$ ($i=1,\cdots,m$) on the right factors.
By the same reason, $\bigtriangleup\{Q_{i}\}$
should provide $\{N_{j}\}$ ($j=1,\cdots,n$)
on their left factors. Therefore, there are
only three meaningful possibilities of $\{Q_{i}\}$
as follows.

\begin{itemize}
\item $\{Q_{k}\}=\{((N)_{J}
\hookrightarrow_{(i_{J},\iota_{i_{J}})}M_{a})\}$,
where $J\subset[n]$. Then $\bigtriangleup\{Q_{i}\}$
will contain the term 

$$
\{(N)_{J}\}\otimes\{M_{a}\}.
$$
\item $\{Q_{k}\}=\{M_{a}\}$, then

$$
\bigtriangleup\{Q_{k}\}=0\otimes\{M_{a}\}+\cdots.
$$
\item $\{Q_{k}\}=\{(N)_{J}\}$
for some $J\subset[n]$, then

$$
\bigtriangleup\{Q_{i}\}=
\{(N)_{J}\}\otimes 0+\cdots.
$$
\end{itemize}
The previous discussions show that the suitable
choices of $\{Q\}$ should be of the following form:

$$
\{Q\}=\{(N)_{\varLambda^{c}}\}\oplus
(\bigoplus\limits_{i\in I}\{((N)_{J_{i}}
\hookrightarrow_{\{(i_{J_{i}},\iota_{i_{J_{i}}})\}}M_{i})\}) 
\oplus\{(M)_{I^{c}}\},
$$
where $\Lambda\subset[n]$,
$\varLambda^{c}=[n]\setminus\varLambda$,
$\{J_{i}\}_{i\in I}\in\mathbf{Part}(\varLambda)$.
Comparing above expression with the formula (3.6), 
we know that $\{Q\}$ should
be taken to be

$$
\{Q\}=\left\{
\begin{array}{cc}
\{(N)_{\varLambda^{c}}\}\oplus
\{((N)_{\varLambda}\hookrightarrow_
{(i_{\varLambda},\iota_{i_{\varLambda}})}
(M)_{[m]})\},
 & \varLambda\not=\emptyset, \\
\{(N)_{[n]}\}\oplus
\{(M)_{[m]}\}. &
\end{array}\right.
$$

Above discussions mean that the formula (4.5) is valid.

\end{proof}

The formula (4.5) suggests us to define a new multiplication
on $\mathcal{H}_{adj}$.

\begin{definition}
Let $\{M_{i}\}$, $\{N_{j}\}$ be connected
($i=1,\cdots,m,\,j=1,\cdots,n$). We define the
multiplication $\bullet$ between 
$\{M_{1}\}\oplus\cdots\oplus\{M_{m}\}$
and $\{N_{1}\}\oplus\cdots\oplus\{N_{n}\}$ as follows:

\begin{equation}
\begin{array}{c}
\{(N)_{[n]}\}\bullet\{(M)_{[m]}\} \\
=\sum\limits_{\varLambda\subset[n],
\varLambda\not=\emptyset} \,\,
\sum\limits_{(i_{\varLambda},\iota_{i_{\varLambda}})}
\{(N)_{\varLambda^{c}}\}\oplus
\{((N)_{\varLambda}\hookrightarrow_
{(i_{\varLambda},\iota_{i_{\varLambda}})}
(M)_{[m]})\} \\
+(\{(N)_{[n]}\}\oplus\{(M)_{[m]}\}),
\end{array}
\end{equation}
where $\Lambda^{c}=[n]\setminus\varLambda$.
\end{definition}

It is easy to see the multiplication (4.6) is
non-commutative. We want to prove the associativity 
of the product $\bullet$.

\begin{theorem}
Let $\{M_{i}\}$, $\{N_{j}\}$ and $\{Q_{k}\}$
be connected ($i=1,\cdots,m,\,j=1,\cdots,n,\,
k=1,\cdots,q$). Then we have

\begin{equation}
\{(N)_{[n]}\}\bullet
(\{(M)_{[m]}\}\bullet\{(Q)_{[q]}\}) 
=(\{(N)_{[n]}\}\bullet\{(M)_{[m]}\})\bullet
\{(Q)_{[q]}\}.
\end{equation}
\end{theorem}

\begin{proof}
The sum on the left side of (4.6) is over all
possible insertion. Therefore, to prove the formula 
(4.7) we need to know what types of the terms
will appear on both sides of (4.7).

$\mathbf{The\,\,situation\,\,of\,\,the\,\,right\,\,side:}$

First, we consider the right side of (4.7). 
By the formulas (4.6), (3.6), we know that 

$$
\begin{array}{c}
(\{(N)_{[n]}\}\bullet\{(M)_{[m]}\})\bullet
\{(Q)_{[q]}\} \\
=\sum\limits_{\varLambda\subset[n],\,
\varLambda\not=\emptyset} \,\,
\sum\limits_{(i_{\varLambda},
\iota_{i_{\varLambda}})}
(\{(N)_{\varLambda^{c}}\}\oplus
\{((N)_{\varLambda}\hookrightarrow_
{(i_{\varLambda},\iota_{i_{\varLambda}})}
(M)_{\varGamma})\}
\oplus\{(M)_{\varGamma^{c}}\}) 
\bullet\{(Q)_{[q]}\} \\
+(\{(N)_{[n]}\}\oplus
\{(M)_{[m]}\})\bullet
\{(Q)_{[q]}\}.
\end{array}
$$

We focus 
on the terms with the following form,

$$
\begin{array}{c}
(\{(N)_{\varLambda^{c}}\}\oplus
\{((N)_{\varLambda}\hookrightarrow_
{(i_{\varLambda},\iota_{i_{\varLambda}})}
(M)_{\varGamma})\}
\oplus\{(M)_{\varGamma^{c}}\}) 
\bullet\{(Q)_{[q]}\}.
\end{array}
\quad(\ast\ast\ast)
$$
In the expression ($\ast\ast\ast$) 

$$
\varGamma=\{i\in[m]|\,\exists j\in[m]
\,\,s.t.\,\,i_{j}\in M_{i}\},
$$
$\varLambda\not=\emptyset$, thus
$\varGamma\not=\emptyset$.

With the same reason due to the formula (4.6),
we have

$$
\begin{array}{c}
(\{(N)_{\varLambda^{c}}\}\oplus
\{((N)_{\varLambda}\hookrightarrow_
{(i_{\varLambda},\iota_{i_{\varLambda}})}
(M)_{\varGamma})\}\oplus\{(M)_{\varGamma^{c}}\}) 
\bullet\{(Q)_{[q]}\} \\
=\sum\limits_{\varLambda_{c,2},
\varLambda_{2},\varGamma_{2},\varGamma_{c,2}}
\{(N)_{\varLambda_{c,1}}\} 
\oplus\{(M)_{\varGamma_{c,1}}\} 
\oplus\{(Q)_{\varXi^{c}}\}  \\
\oplus\{((N)_{\varLambda_{1}}\hookrightarrow_
{(i_{\varLambda_{1}},
	\iota_{i_{\varLambda_{1}}})}
(M)_{\varGamma_{1}})\} 
\oplus   
\sum\limits_{\{\ast\}\cup\{\ast\}\cup\{\ast\}}
\{O_{\varLambda_{c,2},
\varLambda_{2},\varGamma_{2},
\varGamma_{c,2},
\varXi,\{\ast\}\cup\{\ast\}\cup\{\ast\}}\} \\
+\{(N)_{\varLambda^{c}}\}\oplus
\{((N)_{\varLambda}\hookrightarrow_
{(i_{\varLambda},\iota_{i_{\varLambda}})}
(M)_{\varGamma})\}\oplus\{(M)_{\varGamma^{c}}\} 
\oplus\{(Q)_{[q]}\},
\end{array}
$$
where $\varLambda_{c,2}\cup
\varLambda_{2}\cup\varGamma_{2}\cup
\varGamma_{c,2}\not=\emptyset$, and

$$
\begin{array}{c}
\{O_{\varLambda_{c,2},
\varLambda_{2},\varGamma_{2},
\varGamma_{c,2},
\varXi,\{\ast\}\cup\{\ast\}\cup\{\ast\}}\} \\
=\{[((N)_{\varLambda_{c,2}}
\oplus((N)_{\varLambda_{2}} 
\hookrightarrow_
{(i_{\varLambda_{2}},
\iota_{i_{\varLambda_{2}}})}
(M)_{\varGamma_{2}}) \\
\oplus(M)_{\varGamma_{c,2}})
\hookrightarrow_{\{\ast\}\cup\{\ast\}\cup\{\ast\}}
(Q)_{\varXi}]\},
\end{array}
$$
moreover,

\begin{itemize}
\item $\varLambda=\varLambda_{1}
\cup\varLambda_{2}$,
$\varLambda_{1}\cap\varLambda_{2}
=\emptyset$.
\item $\varLambda^{c}=\varLambda_{c,1}
\cup\varLambda_{c,2}$,
$\varLambda_{c,1}\cap\varLambda_{c,2}
=\emptyset$.
\item $\varGamma=\varGamma_{1}
\cup\varGamma_{2}$,$\varGamma_{1}\cap
\varGamma_{2}=\emptyset$.
\item $\varGamma^{c}=\varGamma_{c,1}
\cup\varGamma_{c,2}$,
$\varGamma_{c,1}\cup\varGamma_{c,2}=\emptyset$.
\item $[q]=\varXi\cup\varXi^{c}$, where
the choice of $\varXi$ depends on the
other decompositions mentioned above.
\end{itemize}

Addtionlly, we need to consider the term
$(\{(N)_{[n]}\}\oplus\{(M)_{[m]}\})
\bullet\{(Q)_{[q]}\}$. With the same reason
as above, we have,

$$
\begin{array}{c}
(\{(N)_{[n]}\}\oplus\{(M)_{[m]}\})
\bullet\{(Q)_{[q]}\} \\
=\sum\limits_{\varLambda\subset[n],
\varGamma\subset[m],
\varLambda\cup\varGamma\not=\emptyset}
\{(N)_{\varLambda^{c}}\}\oplus\{(M)_{\varGamma^{c}}\}
\oplus\{(Q)_{\varXi^{c}}\}\oplus \\
\{(((N)_{\varLambda}\oplus(M)_{\varGamma})
\hookrightarrow_{(q_{\varLambda},\iota_{q_{\varLambda}})
\cup(q_{\varGamma},\tau_{q_{\varGamma}})}(Q)_{\varXi})\} \\
+\{(N)_{[n]}\}\oplus\{(M)_{[m]}\}
\oplus\{(Q)_{[q]}\}.
\end{array}
$$
In summary, we get a general expression of
the right side of the formula (4.7) as
follows.

$$
\begin{array}{c}
(\{(N)_{[n]}\}\bullet\{(M)_{[m]}\})\bullet
\{(Q)_{[q]}\} \\
=\sum\limits_{\ast}\{(N)_{\varLambda_{1}}\}
\oplus\{(M)_{\varGamma_{1}}\} 
\oplus\{(Q)_{\varXi^{c}}\}  
\oplus\{((N)_{\varLambda_{2}}\hookrightarrow_
{(i_{\varLambda_{2}},
\iota_{i_{\varLambda_{2}}})}
(M)_{\varGamma_{2}})\} \\
\oplus\{[((N)_{\varLambda_{3}}
\oplus((N)_{\varLambda_{4}} 
\hookrightarrow_
{(i_{\varLambda_{4}},
\iota_{i_{\varLambda_{4}}})}
(M)_{\varGamma_{3}})
\oplus(M)_{\varGamma_{4}})
\hookrightarrow_{\{\ast\}\cup\{\ast\}\cup\{\ast\}}
(Q)_{\varXi}]\} \\
+\{(N)_{[n]}\}\oplus\{(M)_{[m]}\}
\oplus\{(Q)_{[q]}\},
\end{array}
$$
where the sum is over all possible choices
of $\{\varLambda_{i}\}_{i=1}^{4}$,
$\{\varGamma_{i}\}_{i=1}^{4}$ and $\Xi$,
$\{\varLambda_{i}\}_{i=1}^{4}\in
\mathbf{Part}([n])$,
$\{\varGamma_{i}\}_{i=1}^{4}\in
\mathbf{Part}([m])$, $\varLambda_{i}$
or $\varGamma_{j}$ is allowed to be emptyset
for some $i$ or $j$ ($1\leq i,j\leq 4$), and

$$
(\bigcup\limits_{i=1,2,3}\varLambda_{i})
\cup(\bigcup\limits_{i=1,2,3}\varGamma_{i})
\not=\emptyset.
$$

$\mathbf{The\,\,situation\,\,of\,\,the\,\,left\,\,side:}$

We now consider the left side of the formula
(4.7). Similarly, we need to focus on the 
terms with the following form,

$$
\{(N)_{[n]}\}\bullet 
(\{(M)_{\varGamma^{c}}\}
\oplus
\{((M)_{\varGamma}\hookrightarrow_
{(q_{\varGamma},\kappa_{q_{\varGamma}})}
(Q)_{\varXi})\}
\oplus(Q)_{\varXi^{c}}),
\quad(\ast\ast)
$$
where $\varGamma\subset[m]$,
$\varGamma\not=\emptyset$,

$$
\varXi=\{k\in[q]|\,\exists i\in\varGamma,
\,\,s.t.\,\,q_{i}\in Q_{k}\}.
$$
Precisely, $\varXi$ results in a decomposition
of $\varGamma$, $\{I_{k}\}_{k\in\varXi}\in
\mathbf{Part}(\varGamma)$, such that

$$
\begin{array}{c}
\{((M_{i})_{i\in\varGamma}\hookrightarrow_
{\{q_{i},\kappa_{q_{i}}\}_{i\in\varGamma}}
(\bigoplus\limits_{k\in\varXi}Q_{k}))\} \\
=\bigoplus\limits_{k\in\varXi}\{((M_{i})_{i\in I_{k}}
\hookrightarrow_{\{(q_{i},\kappa_{q_{i}})\}_{i\in I_{k}}}Q_{k})\},
\end{array}
$$
where $I_{k}=\{i\in\varGamma|q_{i}\in Q_{k}\}$.

Now we give a description of the expression ($\ast\ast$)
in detail based on the formula (4.6).
Due to the formula (4.6), we have

$$
\begin{array}{c}
\{(N)_{[n]}\}\bullet 
(\{(M)_{\varGamma^{c}}\}
\oplus
\{((M)_{\varGamma}\hookrightarrow_
{(q_{\varGamma},\kappa_{q_{\varGamma}})}
(Q)_{\varXi})\}
\oplus(Q)_{\varXi^{c}}) \\
=\sum\limits_{\varLambda\subset[n],
\varLambda\not=\emptyset}
\{(N)_{\varLambda^{c}}\}\oplus
\sum\limits_{(a_{\varLambda},\iota_{a_{\varLambda}})}
\{((N)_{\varLambda}\hookrightarrow_
{(a_{\varLambda},\iota_{a_{\varLambda}})}
((M)_{\varGamma^{c}}\oplus 
((M)_{\varGamma}\hookrightarrow_
{(q_{\varGamma},\kappa_{q_{\varGamma}})}
(Q)_{\varXi}) \\
\oplus(Q)_{\varXi^{c}}))\}+
\{(N)_{[n]}\}\oplus 
\{(M)_{\varGamma^{c}}\}\oplus
\{((M)_{\varGamma}\hookrightarrow_
{(q_{\varGamma},\kappa_{q_{\varGamma}})}
(Q)_{\varXi})\}
\oplus(Q)_{\varXi^{c}}.
\end{array}
$$
We focus on the term

$$
\{((N)_{\varLambda}\hookrightarrow_
{(a_{\varLambda},\iota_{a_{\varLambda}})}
((M)_{\varGamma^{c}}\oplus 
((M)_{\varGamma}\hookrightarrow_
{(q_{\varGamma},\kappa_{q_{\varGamma}})}
(Q)_{\varXi}) 
\oplus(Q)_{\varXi^{c}}))\}.
$$
We divide $\varLambda$ into three subsets 
$\varLambda_{N\hookrightarrow M}$, 
$\varLambda_{N\hookrightarrow M\hookrightarrow Q}$
and $\varLambda_{N\hookrightarrow Q}$ 
($\{\varLambda_{N\hookrightarrow M},
\varLambda_{N\hookrightarrow M\hookrightarrow Q},
\varLambda_{N\hookrightarrow Q}\}
\in\mathbf{Part}(\varLambda)$) such that
the above term can be divided into three parts.

$$
\begin{array}{c}
\{((N)_{\varLambda}\hookrightarrow_
{(a_{\varLambda},\iota_{a_{\varLambda}})}
((M)_{\varGamma^{c}}\oplus 
((M)_{\varGamma}\hookrightarrow_
{(q_{\varGamma},\kappa_{q_{\varGamma}})}
(Q)_{\varXi}) 
\oplus(Q)_{\varXi^{c}}))\} \\
=\{((N)_{\varLambda_{N\hookrightarrow M}}\hookrightarrow_
{(i_{\varLambda_{N\hookrightarrow M}},
\iota_{i_{\varLambda_{N\hookrightarrow M}}})}
(M)_{\varGamma^{c}})\}\oplus \\
\{((N)_{\varLambda_{N\hookrightarrow M\hookrightarrow Q}}
\hookrightarrow_{(a_{\varLambda_{N\hookrightarrow M\hookrightarrow Q}},
\lambda_{a_{\varLambda_{N\hookrightarrow M\hookrightarrow Q}}})}
((M)_{\varGamma}\hookrightarrow_
{(q_{\varGamma},\kappa_{q_{\varGamma}})}
(Q)_{\varXi})\} \\
\oplus
\{((N)_{\varLambda_{N\hookrightarrow Q}}\hookrightarrow_
{(q_{\varLambda_{N\hookrightarrow Q}},
\iota_{q_{\varLambda_{N\hookrightarrow Q}}})}
(Q)_{\varXi^{c}})\}.
\end{array}
$$
Furthermore, by the formula (3.6) we have:

\begin{itemize}
\item

$$
\begin{array}{c}
\{((N)_{\varLambda_{N\hookrightarrow M}}\hookrightarrow_
{(i_{\varLambda_{N\hookrightarrow M}},
\iota_{i_{\varLambda_{N\hookrightarrow M}}})}
(M)_{\varGamma^{c}})\} \\
=\{((N)_{\varLambda_{N\hookrightarrow M}}\hookrightarrow_
{(i_{\varLambda_{N\hookrightarrow M}},
\iota_{i_{\varLambda_{N\hookrightarrow M}}})}
(M)_{\varGamma_{c,N\hookrightarrow M}})\} 
\oplus\{(M)_{\varGamma_{c,M}}\},
\end{array}
$$
where $\varGamma^{c}=\varGamma_
{c,N\hookrightarrow M}\cup\varGamma_{c,M}$,
$\varGamma_{c,N\hookrightarrow M}\cap\varGamma_{c,M}
=\emptyset$,

\item

$$
\begin{array}{c}
\{((N)_{\varLambda_{N\hookrightarrow Q}}\hookrightarrow_
{(q_{\varLambda_{N\hookrightarrow Q}},
\iota_{q_{\varLambda_{N\hookrightarrow Q}}})}
(Q)_{\varXi^{c}})\} \\
=\{((N)_{\varLambda_{N\hookrightarrow Q}}\hookrightarrow_
{(q_{\varLambda_{N\hookrightarrow Q}},
\iota_{q_{\varLambda_{N\hookrightarrow Q}}})}
(Q)_{\varXi_{c,N\hookrightarrow Q}})\} 
\oplus\{(Q)_{\varXi_{c,Q}}\},
\end{array}
$$
where $\varXi=\varXi_{c,N\hookrightarrow Q}\cup
\varXi_{c,Q}$, $\varXi_{c,N\hookrightarrow Q}\cap
\varXi_{c,Q}=\emptyset$.

\item Recalling proposition 3.6 and remark 3.2 we have

$$
\begin{array}{c}
\{((N)_{\varLambda_{N\hookrightarrow M\hookrightarrow Q}}
\hookrightarrow_{(a_{\varLambda_{N\hookrightarrow M\hookrightarrow Q}},
\lambda_{a_{\varLambda_{N\hookrightarrow M\hookrightarrow Q}}})}
((M)_{\varGamma}\hookrightarrow_
{(q_{\varGamma},\kappa_{q_{\varGamma}})}(Q)_{\varXi})\} \\
=\{((N)_{\varLambda^{(1)}}
\oplus((N)_{\varLambda^{(2)}}
\hookrightarrow_{(i_{\varLambda^{(2)}}^{\prime},
\iota_{i_{\varLambda^{(2)}}^{\prime}}^{\prime})}(M)_{\varGamma^{(2)}})
\oplus(M)_{\varGamma^{(1)}} \\
\hookrightarrow_{(\ast\ast)\cup(\ast\ast)\cup(\ast\ast)}
(Q)_{\varXi})\}
\end{array}
$$

\end{itemize}

Additionally, we need to consider the term
$\{(N)_{[n]}\}\bullet
(\{(M)_{[m]}\}\oplus\{(Q)_{[q]}\})$

$$
\begin{array}{c}
\{(N)_{[n]}\}\bullet
(\{(M)_{[m]}\}\oplus\{(Q)_{[q]}\}) \\
=\sum\limits_{\varLambda\subset[n],
\varLambda\not=\emptyset}
\{(N)_{\varLambda^{c}}\}\oplus
\{(N)_{\varGamma^{c}}\}\oplus
\{(N)_{\varXi^{c}}\}\oplus
\{((N)_{\varLambda}\hookrightarrow_
{(a_{\varLambda},\iota_{a_{\varLambda}})}
((M)_{\varGamma}\oplus(Q)_{\varXi}))\} \\
+\{(N)_{[n]}\}\oplus
\{(M)_{[m]}\}\oplus\{(Q)_{[q]}\}.
\end{array}
$$

In summary, we know that the left side
of the formula (4.7) has same form as
the one of the right side. Noting that the 
expressions on the both sides are the sum
and direct sum for all possible insertion,
thus the formula (4.7) is valid.

\end{proof}

With the help of theorem 3.2, the discussions
concerning the product on $\mathcal{H}_{adj}^{\ast}$
can be reduced to the situation of $\mathcal{H}_{adj}$.
In our setting, we do not distinguish the zero matrix
with different order. By definition 3.1 we have

$$
\{0\}\bullet\{M\}=\{M\}\bullet\{0\}=\{M\}.
$$
Thus $(\mathcal{H}_{adj},\bullet,\{0\})$
is an unital algebra over $\mathbb{K}$. We define
a map $\mathcal{M}$ from $(\mathcal{H}_{adj},\bullet,\{0\})$ 
to $(\mathcal{H}_{adj}^{\ast},\bullet,\eta)$ as follows:

\begin{equation}
\mathcal{M}:\{M_{1}\}\oplus\cdots\oplus\{M_{m}\}
\mapsto f_{\{M_{1}\}\oplus\cdots\oplus\{M_{m}\}},\,
\mathcal{M}:\{0\}\mapsto\eta.
\end{equation}
In (4.8) each $\{M_{i}\}$ is connected ($i=1,\cdots,m$).

From definition 4.1, theorem 4.1 and theorem 4.2
we immidiately have the conclusion about $\mathcal{M}$.

\begin{proposition}
The map $\mathcal{M}$ defined by (4.8) is an algebraic
isomorphism from $(\mathcal{H}_{adj},\bullet,\{0\})$ to
$(\mathcal{H}_{adj}^{\ast},\bullet,\eta)$.
\end{proposition}

By definition of $\mathcal{H}_{adj}$, we know that
$M_{adj}(+\infty,\mathbb{N})$ plays the role of
the bases in $\mathcal{H}_{adj}$. On the other hand,
we know that

$$
M_{adj}(+\infty,\mathbb{N})=
\{\bigoplus\limits_{i=1}^{m}\{M_{i}\}|
\,m\in\mathbb{N},\,\{M_{i}\}\in 
M_{adj}(m_{i},\mathbb{N})\diagup\sim\,\,is\,\,
connected,\,\,1\leq i\leq m\}.
$$
Thus, the formula (4.2) suggests us to define a new
coproduct on $\mathcal{H}_{adj}$ in the following way.

\begin{definition}
Let $\{M\}=\bigoplus_{i=1}^{m}\{M_{i}\}$,
where each $\{M_{i}\}\in M_{adj}(m_{i},\mathbb{N})$
is connected ($i=1,\cdots,m$). Then we define
the coproduct to be

\begin{equation}
\bigtriangleup_{1}\{M\}=\{M\}\otimes\{0\}+
\{0\}\otimes\{M\} 
+\sum\limits_{I\subset[m],\,
I,I^{c}\not=\emptyset}\{(M)_{I}\}\otimes
\{(M)_{I^{c}}\},
\end{equation}
where $I^{c}=[m]\setminus I$.
Particularly, $\bigtriangleup_{1}\{0\}
=\{0\}\otimes\{0\}$.
\end{definition}

The product $\bullet$ can be exteneded to
the situation of $\mathcal{H}_{adj}\otimes\mathcal{H}_{adj}$.
Let $(M)_{[m]},(N)_{[n]},
(Q)_{[q]}, \\ (R)_{[r]}\in
M_{adj}(+\infty,\mathbb{N})$, we define

$$
((M)_{[m]}\otimes(N)_{[n]})\bullet
((Q)_{[q]}\otimes(R)_{[r]})
=((M)_{[m]}\bullet(Q)_{[q]})\otimes
((N)_{[n]}\bullet(R)_{[r]}).
$$ 
It is easy to check that the product defined above
is well defined.
 
It is obvious that $\bigtriangleup_{1}$
is co-commutative. Firstly, we will prove
$\bigtriangleup_{1}$ is co-associative.

\begin{theorem}
We have

\begin{equation}
(\bigtriangleup_{1}\otimes 1)\bigtriangleup_{1}=
(1\otimes\bigtriangleup_{1})\bigtriangleup_{1}.
\end{equation}
\end{theorem}

\begin{proof}
Let $\{M\}=\bigoplus_{i=1}^{m}\{M_{i}\}$,
where each $\{M_{i}\}\in 
M_{adj}(m_{i},\mathbb{N})\diagup\sim$
is connected ($i=1,\cdots,m$).
By a straightforward calculation, we have

$$
\begin{array}{c}
(\bigtriangleup_{1}\otimes 1)\bigtriangleup_{1}\{M\}=
(1\otimes\bigtriangleup_{1})\bigtriangleup_{1}\{M\} \\
=\sum\limits_{I_{1},I_{2},I_{3}}
\{(M)_{I_{1}}\}\otimes\{(M)_{I_{2}}\}
\otimes\{(M)_{I_{3}}\},
\end{array}
$$
where $I_{1}\cup I_{2}\cup I_{3}=[m]$,
$I_{i}\cap I_{j}=\emptyset$ ($i\not= j$), one or two of
$I_{1}, I_{2}, I_{3}$ may be emptyset.

\end{proof}

The coproduct $\bigtriangleup_{1}$ and product
$\bullet$ are compatible.

\begin{theorem}
Let $\{M_{i}\},\{N_{j}\}\in M_{adj}(+\infty,\mathbb{N})$
be connected ($i=1,\cdots,\,j=1,\cdots,n$).
Then, we have

\begin{equation}
\bigtriangleup_{1}(\{(N)_{[n]}\}
\bullet\{(M)_{[m]}\})=
\bigtriangleup_{1}\{(N)_{[n]}\}
\bullet
\bigtriangleup_{1}\{(M)_{[m]}\}.
\end{equation}
\end{theorem}

\begin{proof}
To prove the formula (4.11), we need
to calculate the both sides of (4.11).

$\mathbf{The\,\,situation\,\,of\,\,the\,\,left\,\,side:}$

Recalling the formula (4.6) we have

$$
\begin{array}{c}
\{(N)_{\underline{n}}\}
\bullet\{(M)_{\underline{m}}\}
=\sum\limits_{\varLambda\subset[n],\,
\varLambda\not=\emptyset}
\sum\limits_{(i_{\varLambda},\iota_{i_{\varLambda}})}
\{(N)_{\varLambda^{c}}\}\oplus
\{((N)_{\varLambda}\hookrightarrow_
{(i_{\varLambda},\iota_{i_{\varLambda}})}
(M)_{\varGamma})\}\oplus\{(M)_{\varGamma^{c}}\}.
\end{array}
$$
Therefore

$$
\begin{array}{c}
\bigtriangleup_{1}(\{(N)_{[n]}\}
\bullet\{(M)_{[m]}\}) \\
=\sum\limits_{\varLambda\subset[n],\,
\varLambda\not=\emptyset}
\sum\limits_{(i_{\varLambda},\iota_{\varLambda})}
\bigtriangleup_{1}\{(N)_{\varLambda^{c}}\}\oplus
\bigtriangleup_{1}\{((N)_{\varLambda}\hookrightarrow_
{(i_{\varLambda},\iota_{\varLambda})}
(M)_{\varGamma})\}\oplus
\bigtriangleup_{1}\{(M)_{\varGamma^{c}}\} \\
=\sum\limits_{\varLambda\subset\underline{n},\,
\varLambda\not=\emptyset}
\sum\limits_{(i_{\varLambda},\iota_{\varLambda})}
(\sum\limits_{\varLambda_{c,1}\subset\varLambda^{c}}
\{(N)_{\varLambda_{c,1}}\}\otimes\{(N)_{\varLambda_{c,2}}\})
\oplus(\sum\limits_{\varGamma_{c,1}\subset\varGamma^{c}}
\{(M)_{\varGamma_{c,1}}\}\otimes\{(M)_{\varGamma_{c,2}}\}) \\
\oplus(\sum\limits_{\varGamma_{1}}
\{((N)_{\varLambda_{1}}\hookrightarrow_
{(i_{\varLambda_{1}},\iota_{\varLambda_{1}})}
(M)_{\varGamma_{1}})\}\otimes
\{((N)_{\varLambda_{2}}\hookrightarrow_
{(i_{\varLambda_{2}},\iota_{\varLambda_{2}})}
(M)_{\varGamma_{2}})\}) \\
=\sum\limits_{\varLambda\subset\underline{n},\,
\varLambda\not=\emptyset}
\sum\limits_{(i_{\varLambda},\iota_{\varLambda})}
\sum\limits_{\varLambda_{c,1}\subset\varLambda^{c}}
\sum\limits_{\varGamma_{c,1}\subset\varGamma^{c}}
(\{(N)_{\varLambda_{c,1}}\}\oplus
\{((N)_{\varLambda_{1}}\hookrightarrow_
{(i_{\varLambda_{1}},\iota_{\varLambda_{1}})}
(M)_{\varGamma_{1}})\} \\
\oplus\{(M)_{\varGamma_{c,1}}\}) 
\otimes(\{(N)_{\varLambda_{c,2}}\}\oplus
\{((N)_{\varLambda_{2}}\hookrightarrow_
{(i_{\varLambda_{2}},\iota_{\varLambda_{2}})}
(M)_{\varGamma_{2}})\}\oplus\{(M)_{\varGamma_{c,2}}\}),
\end{array}
$$
where $\{\varLambda_{c,1},\varLambda_{c,2},
\varLambda_{1},\varLambda_{2}\}
\in\mathbf{Part}([n])$,
$\varLambda_{c,1}\cup\varLambda_{c,2}=\varLambda^{c}$,
$\varLambda_{1}\cup\varLambda_{2}=\varLambda$,
$\{\varGamma_{c,1},\varGamma_{c,2},
\varGamma_{1},\varGamma_{2}\}\in\mathbf{Part}([m])$,
$\varGamma_{c,1}\cup\varGamma_{c,2}=\varGamma^{c}$,
$\varGamma_{1}\cup\varGamma_{2}=\varGamma$.
Recalling the proof of propossition 3.3,

$$
\varGamma_{a}=\{i\in\varGamma|\,\exists j\in\varLambda_{a},
\,\,s.t.\,\,i_{j}\in M_{i}\},\,\,a=1,2,
$$
thus $\varGamma_{a}$ is determined by
$\varLambda_{a}$ ($a=1,2$).

Now we take $\varLambda^{(a)}=\varLambda_{c,a}
\cup\varLambda_{a}$ ($a=1,2$), thus,
$\varGamma^{(a)}=\varGamma_{a}\cup\varGamma_{c,a}$
($a=1,2$). Then we have

$$
\begin{array}{c}
\bigtriangleup_{1}(\{(N)_{[n]}\}
\bullet\{(M)_{[m]}\}) \\
=\sum\limits_{\varLambda^{(1)},\varLambda^{(2)},
\varGamma^{(1)},\varGamma^{(2)},}
(\{(N)_{\varLambda^{(1)}}\}\bullet
\{(M)_{\varGamma^{(1)}}\})\otimes
(\{(N)_{\varLambda^{(2)}}\}\bullet
\{(M)_{\varGamma^{(2)}}\}),
\end{array}
$$
where $\varLambda^{(1)}$ or $\varLambda^{(2)}$
may be emptyset, for example, when
$\varLambda^{(1)}=\emptyset$, we define
$\{(N)_{\varLambda^{(1)}}\}=\{0\}$.

$\mathbf{The\,\,situation\,\,of\,\,the\,\,right\,\,side:}$

By definition 4.2 we have

$$
\bigtriangleup_{1}\{(N)_{[n]}\}
=\sum\limits_{\varLambda\subset[n]}
\{(N)_{\varLambda}\}\otimes\{(N)_{\varLambda^{c}}\},\,\,
\bigtriangleup_{1}\{(M)_{[m]}\}
=\sum\limits_{\varGamma\subset[m]}
\{(M)_{\varGamma}\}\otimes\{(M)_{\varGamma^{c}}\}.
$$
Therefore we have

$$
\begin{array}{c}
\bigtriangleup_{1}\{(N)_{[n]}\}\bullet
\bigtriangleup_{1}\{(M)_{[m]}\} \\
=\sum\limits_{\varLambda\subset[n],
\varGamma\subset[m]}
(\{(N)_{\varLambda}\}\bullet\{(M)_{\varGamma}\})\otimes
(\{(N)_{\varLambda^{c}}\}\bullet\{(M)_{\varGamma^{c}}\}).
\end{array}
$$
Comparing the expressions on the both 
sides of (4.11), we know that the formula
(4.11) is valid.

\end{proof}

Recalling the contents in section 2, we know that
the tuple $(\mathcal{H}_{adj},\oplus,\{0\},
\bigtriangleup,\eta)$ is a bialgebra. It is easy
to check that the tuple $(\mathcal{H}_{adj},\bullet,\{0\},
\bigtriangleup_{1},\eta)$ is also a bialgebra.
We consider the reduced coproduct 
$\overline{\bigtriangleup_{1}}$,

$$
\overline{\bigtriangleup_{1}}\{M\}=
\bigtriangleup_{1}\{M\}-\{M\}\otimes\{0\}
-\{0\}\otimes\{M\},\,\{M\}\in M_{adj}(+\infty,\mathbb{N}),
\{M\}\not=\{0\}.
$$
Due to the formula (4.2), there is a obvious 
conclusion as follows.

\begin{proposition}
For each $\{M\}\in M_{adj}(+\infty,\mathbb{N})$
($\{M\}\not=\{0\}$), there is a positive integer
$k$ such that

$$
\overline{\bigtriangleup_{1}}^{k}\{M\}=\{0\},
$$
where 

$$
\begin{array}{ccc}
\overline{\bigtriangleup_{1}}^{k+1}=
(\overline{\bigtriangleup_{1}}\otimes &
\underbrace{1\otimes\cdots\otimes 1} &)
\overline{\bigtriangleup_{1}}^{k}. \\
 & k-times &
\end{array}
$$
\end{proposition}

Proposition 4.4 means that
$(\mathcal{H}_{adj},\bullet,\{0\},
\bigtriangleup_{1},\eta)$ is a conilpotent
bialgebra, therefore, a Hopf algebra.
Similar to the situation of $\mathcal{H}_{adj}^{\ast}$,
the formula (4.2) of the coproduct $\bigtriangleup_{1}$
shows that $\{M\}\in M_{adj}(+\infty,\mathbb{N})$
is connected if and only if

$$
\bigtriangleup_{1}\{M\}=\{M\}\otimes\{0\}
+\{0\}\otimes\{M\}.
$$
Therefore, we have

$$
\mathbf{P}(\mathcal{H}_{adj})
=\mathbf{Span}_{\mathbb{C}}
\{\{M\}\in M_{adj}(+\infty,\mathbb{N})|
\,\{M\}\,\,is\,\,connected\},
$$
where $\mathbf{P}(\mathcal{H}_{adj})$ denotes
the set of the all primitive elements of
$(\mathcal{H}_{adj},\bullet,\{0\},
\bigtriangleup_{1},\eta)$. Let $\{M\},\{N\}\in
M_{adj}(+\infty,\mathbb{N})$ be connected,
then the product $\bullet$ induces a Lie bracket
as follows,

\begin{equation}
[\{M\},\{N\}]=\{M\}\bullet\{N\}
-\{N\}\bullet\{M\}.
\end{equation}
By the formula (4.4) we have

\begin{equation}
[\{M\},\{N\}]=\sum\limits_
{(j,\tau_{j})}\{(M\hookrightarrow_{(j,\tau_{j})}N)\}-
\sum\limits_{(i,\iota_{i})}
\{(N\hookrightarrow_{(i,\iota_{i})}M)\}.
\end{equation}
The formula (4.13) implies that
$[\{M\},\{N\}]\in\mathbf{P}(\mathcal{H}_{adj})$
for $\{M\},\{N\}\in\mathbf{P}(\mathcal{H}_{adj})$.
Hence $\mathbf{P}(\mathcal{H}_{adj})$ is
a Lie algebra. According to Milnor-Moore
theorem (see ?) we know that 

$$
\mathcal{H}_{adj}\cong 
U(\mathbf{P}(\mathcal{H}_{adj})),
$$
i.e. as a Hopf algebra, 
$(\mathcal{H}_{adj},\bullet,\{0\},
\bigtriangleup_{1},\eta)$ is isomorphic
to the enveloping algebra of
$\mathbf{P}(\mathcal{H}_{adj})$,
$U(\mathbf{P}(\mathcal{H}_{adj}))$.
Actually, with the help of the formula
(4.6), we can directly prove that
$\{(M)_{[m]}\}$ can be expressed
by a polynormial of the elements in
$\mathbf{P}(\mathcal{H}_{adj})$.
Precisely, let $\{(M)_{[m]}\}=
\bigoplus_{i=1}^{m}\{M_{i}\}$, each $\{M_{i}\}$
is connected ($i=1,\cdots,m$). Then,
by induction on $m$, we can prove that
$\bigoplus_{i=1}^{m}\{M_{i}\}$ can be
expressed as a polynoremial of $\{M_{i}\}$
($i=1,\cdots,m$) and their insertions 
under the multiplication $\bullet$.

\begin{remark}
Based on the correspondence between the
adjacency matrices and Feynman diagrams,
the Hopf algebra $(\mathcal{H}_{adj},\bullet,\{0\},
\bigtriangleup_{1},\eta)$ means there
is another Hopf algebra structure on the
set of Feynman diagrams induced from the
dual of Connes-Kreimer hopf algebra.
\end{remark}

\end{document}